\pdfoutput=1
\documentclass[10pt,twocolumn,oneside,final]{IEEEtran}

\usepackage{cite}
\usepackage{graphicx}
\usepackage{amsmath}
\usepackage{times}
\usepackage{latexsym}
\usepackage{bm}
\usepackage{amssymb}
\usepackage[center]{caption2}
\usepackage{array}
\usepackage{fancyhdr}
\usepackage{cite,graphicx,amsmath,amssymb}
\usepackage{citesort}
\usepackage{psfrag}
\usepackage{multirow}

\ifCLASSINFOpdf

\else

\fi

\usepackage[dvipsnames,usenames]{color}
\usepackage[usenames,dvipsnames]{color}
\newcommand{\bfbl}[1]{\bf \color{black}#1}
\newcommand{\bl}[1]{\color{black}#1}

\newtheorem{theorem}{Theorem}

\newtheorem{example}{Example}

\newcommand{\tr}{\text{tr}}

\makeatother

\begin{document}
\title{Secure Massive MIMO Transmission
with an Active Eavesdropper}

\author{Yongpeng Wu, \IEEEmembership{Member, IEEE}, Robert Schober, \IEEEmembership{Fellow, IEEE},
 Derrick Wing Kwan Ng, \IEEEmembership{Member, IEEE}, \\ Chengshan Xiao, \IEEEmembership{Fellow, IEEE}, and  Giuseppe Caire,
\IEEEmembership{Fellow, IEEE}.

\thanks{Copyright (c) 2014 IEEE. Personal use of this material is permitted.
However, permission to use this material for any other purposes must be obtained from the IEEE by sending a request to pubs-permissions@ieee.org }

\thanks{This paper was presented in part at IEEE ICC 2015.}

\thanks{The work of Y. Wu and R. Schober was supported by the Alexander von Humboldt Foundation and the German Science Foundation (Grant SCHO831/5-1).
The work of C. Xiao is supported in part by US
National Science Foundation under grant ECCS-1231848. The work of G.
Caire is supported by the Alexander von Humboldt Foundation.}

\thanks{Y. Wu and R. Schober are with Institute for Digital Communications, Universit\"{a}t Erlangen-N\"{u}rnberg,
Cauerstrasse 7, D-91058 Erlangen, Germany (Email: yongpeng.wu@fau.de; robert.schober@fau.de;). }

\thanks{D. W. K. Ng is with the School of Electrical Engineering and
Telecommunications, University of New South Wales, Sydney, N.S.W.,
Australia (E-mail: w.k.ng@unsw.edu.au).}

\thanks{C. Xiao is with the Department of Electrical and Computer Engineering,
Missouri University of Science and Technology, Rolla, MO 65409, USA (Email: xiaoc@mst.edu). }

\thanks{G. Caire is with Institute for Telecommunication Systems, Technical University Berlin, Einsteinufer 25,
10587 Berlin, Germany (Email: caire@tu-berlin.de). }
}

\maketitle

\begin{abstract}
In this paper, we investigate secure and reliable transmission strategies for multi-cell multi-user massive
multiple-input multiple-output (MIMO) systems with a multi-antenna active eavesdropper.
We consider a time-division duplex system where uplink training is required and
an active eavesdropper can attack the training phase to cause
pilot contamination  at the transmitter.
This forces the precoder used in the subsequent downlink transmission phase to
implicitly beamform towards the eavesdropper, thus increasing its received signal power.
Assuming matched filter precoding and artificial noise (AN) generation
at the transmitter, we derive an asymptotic
achievable secrecy rate when the number of transmit antennas approaches infinity.
For the case of a single-antenna active eavesdropper,
we obtain a closed-form expression for
the optimal power allocation policy
for the transmit signal and the AN, and find
the minimum transmit power required
to ensure reliable secure communication.
Furthermore, we show that the transmit antenna correlation diversity of the intended users and the eavesdropper can
 be exploited in order to improve the secrecy rate. In fact, under
 certain orthogonality conditions of the channel covariance matrices,
 the secrecy rate loss introduced
 by the eavesdropper can be completely mitigated.

\end{abstract}

\section{Introduction}
The emergence of smart mobile devices such as smart phones and  wireless
modems has led to an exponential increase in the  demand for wireless data services.
A recent and promising solution to  meet this demand is  massive multiple-input multiple-output (MIMO) technology,
which utilizes {{a very large number of antennas} and simple signal processing at the base station (BS) to serve a comparatively
small (with respective to the number of antennas) number of users. The field of massive MIMO communication
systems was initiated by the pioneering work in \cite{Marzetta2008TWC}
which considered multi-cell multi-user time-division duplex (TDD) communication.
The key idea in \cite{Marzetta2008TWC} is that as the number of transmit antennas
increases,  the effects of uncorrelated receiver noise
and fast fading vanish due to the law of large numbers.  Then, the only residual interference is
caused by the reuse of the same pilot sequences in adjacent cells.
This effect is known as pilot contamination.
Since the publication of \cite{Marzetta2008TWC},
a considerable amount of research has been dedicated to studying  various aspects
of massive MIMO systems \cite{Jose2011TWC,Yin2013JSAC,Adhikary2013TIT,Larsson2014CM,Yin2014JSTSP,Nam2014JSTSP,Adhikary2014JSAC,Sun2015TCOM,Meng2016TCOM}.
In particular, A. Adhikary \textit{et al.} and C. Sun \textit{et al.} design transmission schemes
 to serve different users in orthogonal spatial
resources by exploiting the unique property of the massive MIMO channels in \cite{Adhikary2013TIT} and \cite{Sun2015TCOM}, respectively,
which inspire a null space design for secure massive MIMO transmission in this paper.
Also, massive MIMO has been investigated for various types of systems, including:
Single-cell multi-user uplink/downlink systems \cite{Ngo2013TCOM,Wu2015TWCOM,Lu2016TCOM},
multi-cell multi-user uplink/dowlink systems \cite{Huh2011TIT,Jin2016TVT}, orthogonal
frequency-division multiple access systems \cite{Ng2012TCOM,Wu2014JSTSP,You2016TSP}, non-orthogonal multiple access  systems \cite{Dai2015CM},
and systems employing constant-envelope signals \cite{Mohammed2013TCOM}.

The broadcast nature of the wireless channel makes
it inherently prone to security breaches such as eavesdropping and jamming, which
jeopardizes the privacy of communication in wireless networks.
In order to maintain the required level of privacy, appropriate
signal and information processing techniques have to be employed to ensure reliable and secure
communication.  Traditional approaches to secure communication for preventing unauthorized reception by
eavesdroppers rely on cryptographic encryption implemented in the application layer. These methods
may entail a relatively high complexity due to the required
key distribution and service management\cite{Schneier1998Com}. As a complement
to cryptographic methods, physical layer security, which considers communication
security from an information-theoretic perspective, has attracted significant
research interest recently.  In  Wyner's pioneering work on information-theoretic security,
a ``wiretap channel'' model was defined along with the associated secrecy capacity \cite{Wyner1975BST}.
Wyner's work indicates that the transmitter can reliably send a private message to the receiver,
which cannot be decoded by the eavesdropper,  if the channel of the eavesdropper
is a degraded version of the channel of the desired receiver.
Wyner's result was extended to more general non-degraded channels in \cite{Csiszar1978TIT}.
More recent studies have investigated the capacity and  precoder design for multi-antenna wiretap channels \cite{Khisti2010TIT,Khisti2010TIT_2,Oggier2011TIT,Wu2012TVT,Ng2014TWC,Ng2015TWC,Chen2015CM,Zhu2015EURASIP,Zhang2016TIFS}.
In particular, if only imperfect channel state information (CSI) of the eavesdropper
is available at the transmitter, it is advantageous to transmit
artificial noise (AN)  along with the information-carrying
signal to interfere the decoding process at the eavesdropper \cite{Goel2008TWC,Zhou2010TVT,Huang2012TSP,Wang2015TWC}.
In \cite{Wang2015TSP}, the authors propose a thorough analysis and optimization framework
for artificial noise assisted secure transmission in a MIMO wiretap channel.

Physical layer security for massive MIMO systems with
passive eavesdroppers has been recently studied.
The authors in \cite{Chen2015TWC}
first applied the technique of large-scale antenna array into physical
layer security, and investigate the limiting performance while the number of antennas approached infinity.
Secure massive MIMO transmissions for multi-cell multi-user systems
with imperfect CSI have been investigated in \cite{Zhu2014,Zhu2016TWC},
where a passive eavesdropper attempts to decode the information sent to one of the users.
In \cite{Wang2016TCOM}, a comprehensive performance analysis of AN aided multi-antenna secure transmission in multi-cell
multi-user systems under a stochastic geometry framework is provided.
In \cite{Chen2015TWC,Zhu2014,Zhu2016TWC,Wang2016TCOM}, it was assumed that the channel gains of both the desired receiver and the eavesdropper are
independent and identically distributed (i.i.d.).

Most existing studies on physical layer security assume that perfect CSI of the legitimate channel
is available at the transmitter and do not consider the channel training phase required to acquire the CSI.
However, in TDD communication systems,
the BS needs to estimate the channel for
the subsequent downlink transmission based on pilot sequences sent by the users in an uplink
training phase.  Furthermore, the low rank property of massive MIMO channels has been exploited in \cite{Shen2015CL} to significantly save
the pilot training overhead and reduce the signal processing complexity at users to achieve reliable CSI at the BS.
As a result, a smart eavesdropper might
actively attack this channel training phase by sending the same pilot sequences as the users
to cause pilot contamination at the transmitter, which improves the eavesdropping capability of the eavesdropper significantly \cite{Zhou2012TWC}.

{\bl
Hence, the so-called pilot contamination attack poses a serious secrecy threat to TDD-based massive MIMO systems.
In such systems, the channel hardening due to beamforming with large antenna arrays  \cite{Jose2011TWC} makes the
exploitation of statistical fluctuations due to
fading for secrecy enhancement impossible.
Furthermore, the pilot contamination attack directs  the transmitter beamforming
to the advantage of the eavesdropper. Therefore, for a sufficiently large eavesdropper pilot power, the achievable secrecy rate may
approach zero. This goes against the conventional wisdom  \cite{Larsson2014CM} that massive MIMO inherently facilitates secure communication
because the base station can form very narrow beams focusing on the target users and therefore avoiding  spill over of the
signal power in other directions.}

A single cell massive MIMO system with an active eavesdropper was investigated for i.i.d. fading channels \cite{Im2013,Kapetanovic2013,Basciftci2016}.
However, systematic approaches for
combating the pilot contamination attack of a multi-antenna active eavesdropper and maintaining secrecy of
communication in correlated fading channels
were not provided in \cite{Im2013,Kapetanovic2013,Basciftci2016}
and have not been studied in the literature, yet.

In this paper, we study secure transmission over correlated fading channels in TDD multi-cell multi-user massive MIMO systems in the presence
of a multi-antenna active eavesdropper.
We assume that in the uplink training phase,
the active eavesdropper sends the same pilot sequence as the desired receiver to impair the channel estimation
at the transmitter, i.e., to cause pilot contamination at the transmitter.
Subsequently, the transmitter uses the estimated channel for calculation of
the precoder for downlink transmission. This paper makes the following key contributions:

\begin{enumerate}

 \item {\bl We introduce a pilot contamination precoder, which allows
the eavesdropper to optimize its attack. The proposed pilot contamination precoder is provided
in closed form and
 maximizes the total average estimation error variance of the desired user's channel. }

 \item  We derive a closed-form expression for the asymptotic
achievable secrecy rate for TDD multi-cell multi-user massive MIMO systems
employing matched filter precoding and AN generation (we refer to this design as MF-AN design)
at the transmitter to combat a multi-antenna active eavesdropper.
Based on the derived asymptotic expression,
which is valid if the number of transmit antennas tends to infinity,
the optimal power allocation policy
for the information signal and the AN can be found by a simple one-dimension numerical search.
{\bl Then, we show that, in the presence of an active
eavesdropper,  the secrecy rate
is not a monotonically increasing function of the signal-to-noise ratio (SNR).}
For the special case of a single-antenna eavesdropper,
we obtain the optimal power allocation policy
for the transmit signal and the AN in closed form.
In addition, we obtain
the minimum transmit signal power required
to ensure secure transmission.

\item For the case of correlated fading channels, we reveal that the impact of the active eavesdropper
vanishes when the signal space (i.e., the span of the eigenvectors of the channel correlation matrix
that correspond to non-zero eigenvalues) of the users and the eavesdropper are mutually orthogonal.
Inspired by this observation, we exploit the low rank property of the transmit correlation matrices of massive MIMO channels \cite{Yin2013JSAC,Adhikary2013TIT,Yin2014JSTSP,Nam2014JSTSP,Adhikary2014JSAC,Sun2015TCOM,Shen2015CL}
to design an efficient precoding scheme that transmits in the null space (NS) of the transmit correlation matrix of
the eavesdropper (we refer to this precoding solution as NS design).
{\bl Unlike the conventional NS design for the perfect CSI case \cite{Khisti2010TIT}, the proposed NS design
  can completely remove the impact of the pilot contamination attack by performing joint uplink and downlink processing. }
For the special case of a single-antenna eavesdropper, we derive a
threshold that can be used
to determine whether the MF-AN design or  the NS design
is preferable for given channel and eavesdropper parameters.

\item We propose a unified design which combines the  MF-AN design
and the NS design.  Numerical results indicate that the proposed unified design
can effectively mitigate the pilot contamination attack of an
active eavesdropper in massive MIMO systems.

\end{enumerate}

The remainder of this paper is organized as follows. In Section II,  we introduce the adopted multi-cell multi-user massive MIMO system model
with a multi-antenna active eavesdropper. In Section III,
we derive an expression for the
asymptotic achievable secrecy rate
for the MF-AN design when the number of transmit antennas tends to infinity.
Based on this expression, we derive transmission
strategies  to combat the pilot contamination attack from the active eavesdropper.
In Section IV, we provide  several novel insights
for the single-antenna eavesdropper case. Numerical results are presented  in Section V, and the main results are summarized in Section VI.

\emph{Notation:} Vectors  are denoted by lower-case bold-face letters;
matrices are denoted by upper-case bold-face letters. Superscripts $(\cdot)^{T}$, $(\cdot)^{*}$, and $(\cdot)^{H}$
stand for the matrix transpose, conjugate, and conjugate-transpose operations, respectively. We use  ${\tr}({\bf{A}})$ and ${\bf{A}}^{-1}$
to denote the trace operation and the
inverse of matrix $\bf{A}$, respectively.
$\|\cdot \|$ and $| \cdot |$  denote the Euclidean norm of a matrix/vector and
a scalar, respectively. ${\rm{diag}}\left\{\bf{b}\right\}$ denotes a diagonal matrix
with the elements of vector $\bf{b}$ on its main diagonal.
${\rm vec}\left({\mathbf{A}}\right)$ stacks all columns of matrix $\mathbf{A}$ into a vector.
The $M \times M$ identity matrix is denoted
by ${\bf{I}}_M$, and the all-zero $M \times N$ matrix and $N \times 1$ vector are denoted by $\bf{0}$.
The field of complex numbers is denoted
by $\mathbb{C}$ and $E\left[\cdot\right]$ denotes statistical
expectation.
We use  $\mathbf{x} \sim \mathcal{CN} \left( {\mathbf{0},{{\bf{R}}_N}} \right)$
to denote a circularly symmetric complex Gaussian vector
$\mathbf{x} \in {\mathbb{C}^{N \times 1}}$ with zero mean and covariance matrix ${\bf{R}}_N$.
Based on \cite[Definition II.1]{Shin_1}, we use $\mathbf{X} \sim \mathcal{CN} \left( {\mathbf{0},{{\bf{R}}}_N \otimes {{\bf{R}}}_M }  \right)$
to denote a circularly symmetric complex Gaussian matrix $\mathbf{X} \in {\mathbb{C}^{N \times M}}$
with zero mean and covariance matrix ${{\bf{R}}}_N \otimes {{\bf{R}}}_M$.
$\left\{\mathbf A \right\}_{ij}$ returns
the element of matrix $\mathbf{A}$ in the $i$th row and the $j$th column.
{\bl ${{\bf{1}}_{{N}}} $ denotes an ${N\times 1}$ vector with all elements
equal to $1$.
${\bf{e}}_r$ denotes the unit-vector with a one as the $r$th element and zeros for all other elements.}
${\left[ x \right]^ + }$ stands for $\max \left\{ {0,x} \right\}$,
$\otimes$ denotes the Kronecker product, and $ A \mathop \to \limits^{{N} \to \infty } B$ means that
$A$  converges almost surely to $B$ as $N$ goes to infinity.

\section{System Model} \label{sec:multi}
We consider a multi-cell multi-user system with $L + 1$ cells, cf. Figure \ref{system model}.
Each cell contains a BS equipped
with $N_t$ antennas and $K$ single-antenna users. Without loss of generality,
we denote the reference cell  by
$l = 0$.  An active eavesdropper with $N_e$ antennas
(equivalent to $N_e$ cooperative single-antenna eavesdroppers)
is located in the reference cell.  The eavesdropper seeks to recover the private message intended for a specific target user $m$.

\begin{figure*}[!ht]
\centering
\includegraphics[width=0.8\textwidth]{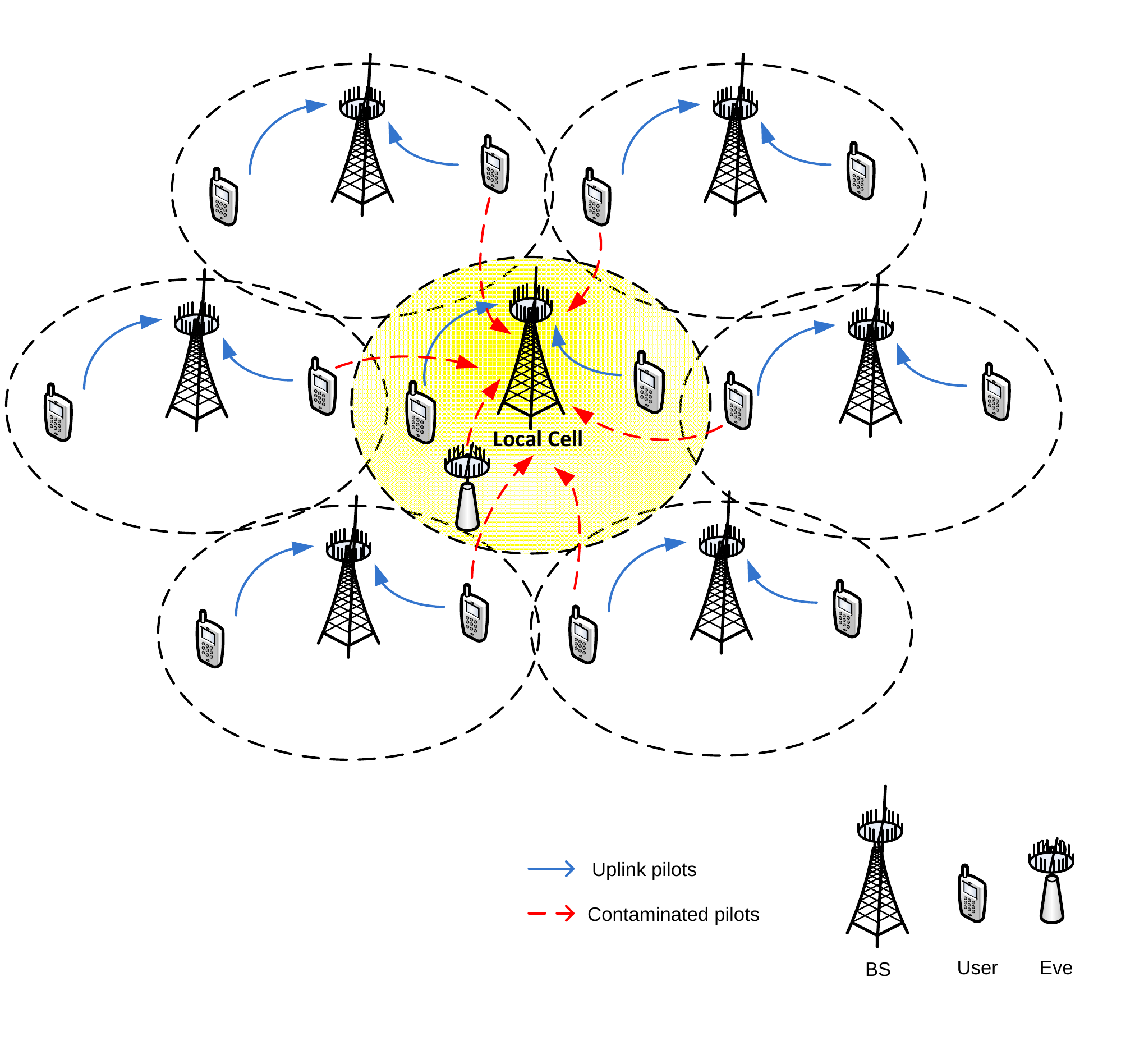}
\caption {\space\space Model of considered multi-cell massive MIMO system.}
\label{system model}
\end{figure*}

\subsection{Uplink Training and Channel Estimation}
In the uplink training and channel estimation phase, the received signal $\mathbf{Y}_0 \in \mathbb{C} {^{N_t  \times \tau}}$ at the BS
in the reference cell is given by \cite{Yin2013JSAC}
\begin{multline} \label{eq:Y_0_model}
{\bf{Y}}_0 = \sum\limits_{k = 1}^K {\sqrt {{P_{0k}}} {\bf{h}}_{0k}^0{\boldsymbol{\omega }}_{0k}^T}  + \sum\limits_{l = 1}^L {\sum\limits_{k = 1}^K {\sqrt {{P_{lk}}}
 {\bf{h}}_{lk}^0{\boldsymbol{\omega }}_{lk}^T} } \\  + \sqrt {\frac{{{P_E}}}{{{N_e}}}} {\bf{{H}}}_{E}^0 {\bfbl P}_{\bl e} {\bf{W}}_e + {\bf{N}}.
\end{multline}
Here, $P_{lk}$ and ${{\boldsymbol{\omega}}_{lk}} \in \mathbb{C} {^{\tau  \times 1}}$ are the average transmit power and the pilot sequence of the $k$th user
in the $l$th cell, where $\tau$ denotes the length of the pilot sequence.
${\bf{h}}_{lk}^p = \left( {\bf{R}}_{lk}^p  \right)^{1/2} \mathbf{g}_{lk}^p  \in \mathbb{C}{^{{N_t} \times 1}}$ denotes
the channel vector between the $k$th user in the $l$th cell and the BS in the $p$th cell,
where $\mathbf{g}_{lk}^p \sim \mathcal{CN} \left( {\mathbf{0},{{\bf{I}}_{{N_t}}}} \right)$ and ${\bf{R}}_{lk}^p  \in {\mathbb{C}^{N_t  \times N_t }}$ is the correlation matrix
of channel ${\bf{h}}_{lk}^p$.  $P_E$ denotes the average transmit power of the eavesdropper when attacking the uplink training.
${\bf{H}}_E^l =  \left({\bf{R}}_{E,T}^l\right)^{1/2}  \mathbf{G}_E^l  \left({\bf{R}}_{E,R}^l\right)^{1/2} \in \mathbb{C}{^{{N_t} \times N_e}} $
denotes the channel between the eavesdropper and the BS in the $l$th cell, where
 $\mathbf{G}_E^l \sim \mathcal{CN}\left( {{{\bf{0}}}, \mathbf{I}_{N_t} \otimes \mathbf{I}_{N_e}} \right)$.
${\bf{R}}_{E,T}^l \in \mathbb{C}{^{{N_t} \times N_t}}$ and ${\bf{R}}_{E,R}^l \in \mathbb{C}{^{{N_e} \times N_e}}$  are the transmit and receive correlation matrices
of channel ${\bf{H}}_E^l$. {\bl The active eavesdropper attacks the channel estimation process of the $m$th user
in the reference cell by sending pilot contamination sequences ${\bfbl P}_{\bl e} {\bf{W}}_e$, where ${\bfbl P}_{\bl e} \in \mathbb{C} {^{N_e  \times N_e }}$
is the pilot contamination precoder and
${\bf{W}}_e = {\left[ {{{\boldsymbol{\omega }}_{0m}},{{\boldsymbol{\omega }}_{0m}}, \cdots ,{{\boldsymbol{\omega }}_{0m}}} \right]^T} \in \mathbb{C} {^{N_e  \times \tau}}$.} $ \mathbf{N} \in {\mathbb{C}^{ N_t \times \tau}}$  is a Gaussian noise matrix
with i.i.d. elements of zero-mean and variance $N_0$.
{\bl The structure of the uplink received signal  at the BS in the local cell is illustrated  in Figure \ref{uplink}.   }

\begin{figure*}[!ht]
\centering
\includegraphics[width=0.9\textwidth]{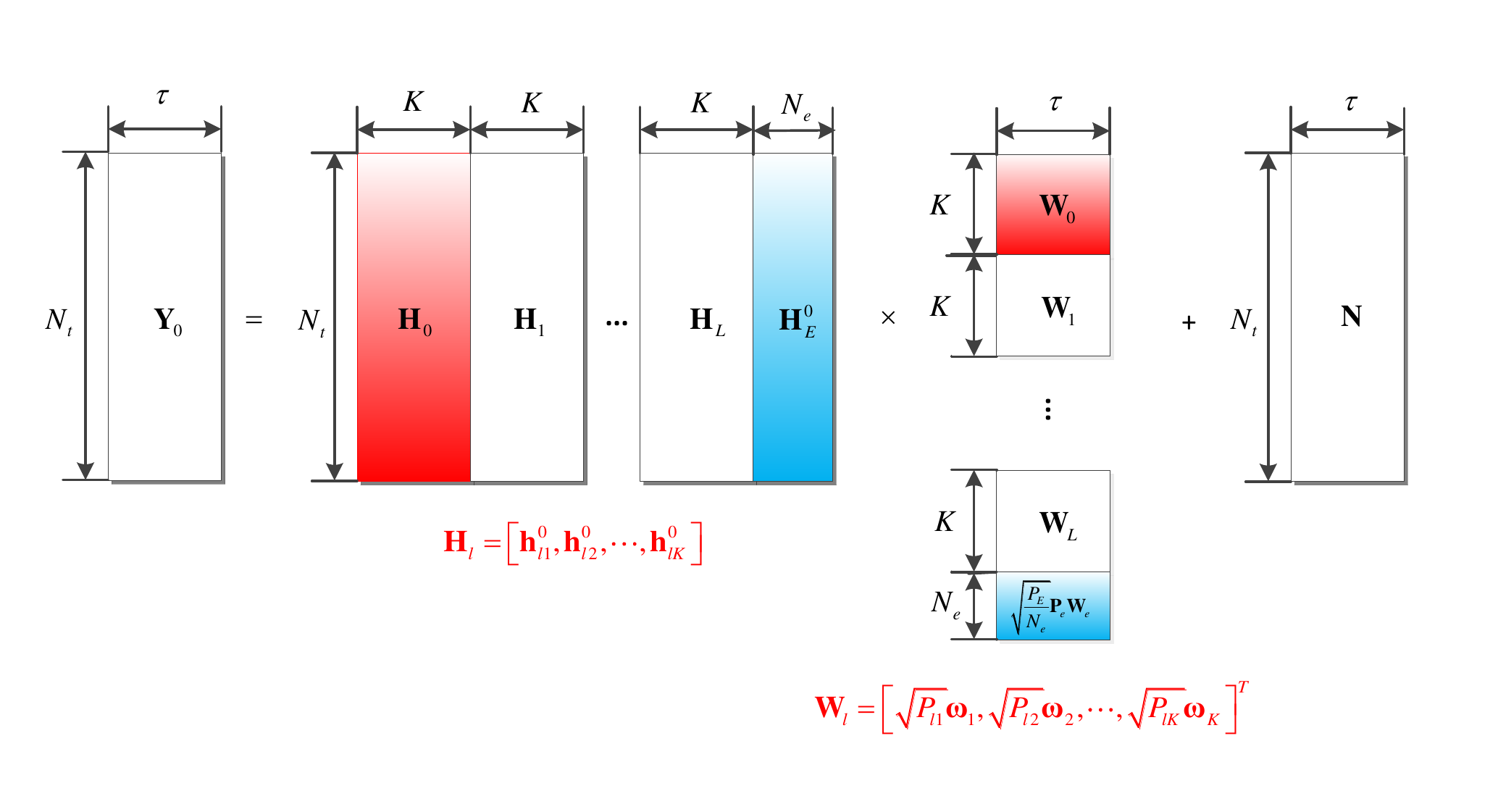}
\caption {\space\space Structure of uplink received signal  at the BS in the local cell.}
\label{uplink}
\end{figure*}

Eq. (\ref{eq:Y_0_model}) can be rewritten as
 \begin{multline} \label{eq:y_0_multi}
{{\bf{y}}_0} = \sum\limits_{k = 1}^K {\sqrt {{P_{0k}}} \left( {{{\boldsymbol{\omega }}_{0k}} \otimes {{\bf{I}}_{{N_t}}}} \right){\bf{h}}_{0k}^0}  \\ + \sum\limits_{l = 1}^L {\sum\limits_{k = 1}^K {\sqrt {{P_{lk}}} \left( {{{\boldsymbol{\omega }}_{lk}} \otimes {{\bf{I}}_{{N_t}}}} \right){\bf{h}}_{lk}^0} } \\
  + \sqrt {\frac{{{P_E}}}{{{N_e}}}} \left( {{{\boldsymbol{\omega }}_{0m}} \otimes {{\bf{I}}_{{N_t}}}} \right)  \sum\limits_{r = 1}^{{N_e}}  {{\bf{h}}_{{\rm eff},r}^0 } + {\bf{n}},
\end{multline}
where ${{\bf{y}}_0} = {\rm vec}\left({\mathbf{Y}}_0\right)$, $\mathbf{n} = {\rm vec}\left({\mathbf{N}}\right)$, and {\bl ${\bf{h}}_{{\rm eff},r}^l$ denotes the $r$th column of matrix ${\bf{H}}_E^l \mathbf{P}_e$},
$l = 0, 1,\cdots,L$.
We assume that the same $K$ orthogonal pilot sequences
are used by the $K$ users in each cell \cite{Marzetta2008TWC}, i.e.,
${{\boldsymbol{\omega }}_{0k}} = {{\boldsymbol{\omega }}_{1k}} =  \cdots ={{\boldsymbol{\omega }}_{Lk}} = {{\boldsymbol{\omega }}_k}$, ${\boldsymbol{\omega }}_{lk}^H{{\boldsymbol{\omega }}_{lk}} = \tau ,
{\boldsymbol{\omega }}_{lk}^H{{\boldsymbol{\omega }}_{lp}} = 0$, $\forall k\neq p$.  Then,
 the minimum mean square error (MMSE) estimate of ${\bf{h}}_{0m}^0$  is given by \cite{Kailath2000}
\begin{multline} \label{eq:mmse_h0m}
{\widehat {\bf{h}}_{0m}^0} = \sqrt {{P_{0m}}} {\bf{R}}_{0m}^0{\left(  {{N_0}{{\bf{I}}_{{N_t}}}  + }
\tau \left( {\sum\limits_{t = 0}^L {{P_{tm}}{\bf{R}}_{tm}^0}  } \right. \right. } \\  \left. { \left. {+ {P_E} {r_{E,R}^0{\bf{R}}_{E,T}^0}  } \right)} \right)^{ - 1}{\widetilde {\bf{y}}_{0m}},
\end{multline}
where
{\bl
 \begin{multline} \label{eq:y_0m}
{\widetilde {\bf{y}}_{0m}}   =  \sqrt {{P_{0m}}} \tau {\bf{h}}_{0m}^0   +  \sum\limits_{t = 1}^L {\sqrt {{P_{tm}}} \tau {\bf{h}}_{tm}^0}
  \\ +  \sqrt {\frac{{{P_E}}}{{{N_e}}}} \tau \sum\limits_{r = 1}^{{N_e}} { {\bf{h}}_{{\rm eff},r}^0 }
+  {\left( {{{\boldsymbol{\omega }}_m}  \otimes  {{\bf{I}}_{{N_t}}}} \right)^H}{\bf{n}}
\end{multline}}and
$r_{E,R}^l = \frac{1}{{{N_e}}}\sum\nolimits_{r = 1}^{{N_e}}
{\sum\nolimits_{s = 1}^{{N_e}} {{{\left\{
{\bl \mathbf{P}_e^H} \bl{{\bf{R}}_{E,R}^l} {\bl  \mathbf{P}_e} \right\}}_{rs}}} }$, $l = 0,1,\cdots,L$.
The actual channel vector ${\bf{h}}_{0m}^0$ can be written
as ${\bf{h}}_{0m}^0 = {\bf{\widehat h}}^0_{0m} + {\bf{e}}_{0m}^0$, where the estimated channel ${\bf{\widehat h}}^0_{0m} \! \sim \! \mathcal{CN} \!\left(\!{\mathbf{0}, {\widehat {\bf{R}}_{0m}^0} }\!\right)$
and the estimation error ${\bf{e}}_{0m}^0 \sim \mathcal{CN} \left({\mathbf{0}, {{\bf{R}}_{0m}^0}  -  {\widehat {\bf{R}}_{0m}^0}}\right)$ are mutually independent. Correlation
matrix ${\widehat {\bf{R}}_{0m}^0}$ is given by
\begin{multline} \label{eq:R0m0_est}
\widehat {\bf{R}}_{0m}^0 =
{P_{0m}}\tau {\bf{R}}_{0m}^0 \\
\times {\left( {{N_0}{{\bf{I}}_{{N_t}}}  + \tau \left( {\sum\limits_{t = 0}^L {{P_{tm}}{\bf{R}}_{tm}^0}   +  {P_E} {r_{E,R}^0{\bf{R}}_{E,T}^0}  } \right)} \right)^{ - 1}}{\bf{R}}_{0m}^0.
\end{multline}
Similarly, the MMSE channel estimates for the $m$th user and the $k$th user, $k = 1,2,\cdots,K, k\neq m$, in the $l$th cell, $l = 0,1,\cdots,L$, are given by
\begin{multline} \label{eq:hnmn_est}
\widehat {\bf{h}}_{lm}^l =
 \sqrt {{P_{lm}}} {\bf{R}}_{lm}^l \\
 \times {\left( {{N_0}{{\bf{I}}_{{N_t}}}  + \tau \left( {\sum\limits_{t = 0}^L {{P_{tk}}{\bf{R}}_{tm}^l}  +  {P_E} {r_{E,R}^l{\bf{R}}_{E,T}^l}   } \right)} \right)^{ - 1}}{\widetilde {\bf{y}}_{lm}},
\end{multline}
 \begin{multline} \label{eq:y_nn}
{\widetilde {\bf{y}}_{lm}} = \sqrt {{P_{lm}}} \tau {\bf{h}}_{lm}^l + \sum\limits_{t = 0, t \neq l }^L {\sqrt {{P_{tm}}} \tau {\bf{h}}_{tm}^l}  \\
+ \sqrt {\frac{{{P_E}}}{{{N_e}}}} \tau \sum\limits_{r = 1}^{{N_e}} {{\bf{h}}_{E,r}^l} + {\left( {{{\boldsymbol{\omega}}_m} \otimes {{\bf{I}}_{{N_t}}}} \right)^H}{\bf{n}}
\end{multline}
and
\begin{align} \label{eq:hnnk_est}
{\widehat {\bf{h}}_{lk}}^l & = \sqrt {{P_{lk}}} {\bf{R}}_{lk}^l{\left( {{N_0}{{\bf{I}}_{{N_t}}} + \tau \sum\limits_{t = 0}^L {{P_{tk}}{\bf{R}}_{tk}^l} } \right)^{ - 1}}{\widetilde {\bf{y}}_{lk}}, \\
{\widetilde {\bf{y}}_{lk}} & = \sqrt {{P_{lk}}} \tau {\bf{h}}_{lk}^l + \sum\limits_{t = 0,t \ne l}^L {\sqrt {{P_{tm}}} \tau {\bf{h}}_{tk}^l}  + {\left( {{{\boldsymbol{\omega }}_k} \otimes {{\bf{I}}_{{N_t}}}} \right)^H}{\bf{n}},
\end{align}
respectively.
Furthermore, the correlation matrices of $\widehat {\bf{h}}_{lm}^l $ and ${\widehat {\bf{h}}_{lk}}^l$  are obtained as
 \begin{multline}
\widehat {\bf{R}}_{lm}^l =
{P_{lm}}\tau {\bf{R}}_{lm}^l \\
\times {\left( {{N_0}{{\bf{I}}_{{N_t}}} + \tau \left( {\sum\limits_{t = 0}^L {{P_{tk}}{\bf{R}}_{tm}^l}  + {P_E} {r_{E,R}^l{\bf{R}}_{E,T}^l}    } \right)} \right)^{ - 1}}{\bf{R}}_{lm}^l \label{eq:Rnmn_est}
\end{multline}
and
\begin{equation}
\widehat {\bf{R}}_{lk}^l = {P_{lk}}\tau {\bf{R}}_{lk}^l{\left( {{N_0}{{\bf{I}}_{{N_t}}} + \tau \sum\limits_{t = 0}^L {{P_{tk}}{\bf{R}}_{tk}^l}  } \right)^{ - 1}}{\bf{R}}_{lk}^l,
\end{equation}
respectively.

{\emph{Remark 1:}} We assume that the correlation matrices ${\bf{R}}_{tm}^l$
of the users
and  ${r_{E,R}^l{\bf{R}}_{E,T}^l}$ of the eavesdropper are perfectly known at the
legitimate transmitter, see (\ref{eq:mmse_h0m}), (\ref{eq:hnmn_est}), (\ref{eq:hnnk_est}).
In this context, we note that for massive MIMO systems, it is reasonable to assume that the
statistical CSI of the users of the system is known at the BS \cite{Yin2013JSAC,Wu2014TSP}.
Hence, our system model is applicable to the case where
the BS attempts to transmit a private message to some users while treating the other users as eavesdroppers, i.e.,
the eavesdropper is an idle user of the system.  Therefore, the
statistical CSI of the eavesdropper can be assumed to be known.
Nevertheless,  the assumption that the
statistical CSI  of the active
eavesdropper is available at the transmitter may also be reasonable if the eavesdropper is
not an idle user. In particular, we can
obtain $E\left[{\widetilde {\bf{y}}_{lm}}{\widetilde {\bf{y}}_{lm}}^H\right]$ by averaging
${\widetilde {\bf{y}}_{lm}}$ over different data slots.
Eq. (\ref{eq:y_nn}) suggests that $E\left[{\widetilde {\bf{y}}_{lm}}{\widetilde {\bf{y}}_{lm}}^H\right]$ is
the sum of the correlation matrices of all users, the eavesdropper, and the noise.
Then, ${r_{E,R}^l{\bf{R}}_{E,T}^l}$ can be obtained by
subtracting the correlation matrices of the legitimate users
and the noise from $E\left[{\widetilde {\bf{y}}_{lm}}{\widetilde {\bf{y}}_{lm}}^H\right]$.

{\bl {\emph{Remark 2:}} For conventional massive MIMO systems with pilot contamination but without active eavesdropping \cite{Marzetta2008TWC,Jose2011TWC},
the term ${P_E} {r_{E,R}^0{\bf{R}}_{E,T}^0}$ is not present in (\ref{eq:mmse_h0m}). In this case,
the BS can employ a user scheduling scheme \cite{Yin2013JSAC,Yin2014JSTSP} to control the
pilot contamination and reduce the impact of ${{P_{tm}}{\bf{R}}_{tm}^0}$ in (\ref{eq:mmse_h0m}).
However, this is not possible for active eavesdropping and the term  ${P_E} {r_{E,R}^0{\bf{R}}_{E,T}^0}$
in (\ref{eq:mmse_h0m}) can not be avoided.
This is an important  difference between conventional massive MIMO systems and the
system considered in this paper.

The eavesdropper can optimize its pilot contamination precoder ${\bf P}_e$
to conduct a best possible attack. Since the precoder is computed based on
 the estimated channel,
the leakage of the desired signal will increase
when the channel estimation error increases.
Therefore, from the eavesdropper's
perspective, the pilot contamination attack should  impair
the accuracy of the channel estimation as much as possible.
 As a result, the eavesdropper should optimize
${\bf P}_e$ such that  the total average
estimation error $\tr \left({\bf{R}}_{0m}^0 - \widehat {\bf{R}}_{0m}^0\right)$ is maximized.
From (\ref{eq:R0m0_est}), we observe that the total average
estimation error $\tr \left({\bf{R}}_{0m}^0 - \widehat {\bf{R}}_{0m}^0\right)$  is a monotonically increasing
function of $r_{E,R}^0$. Hence, maximizing $r_{E,R}^0$ is desirable.
Therefore, finding the pilot contamination precoder ${\bf P}_e$  amounts  to solving the following optimization
problem
\begin{equation} \label{eq:pilot_attack}
\mathop {\max }\limits_{{\bf{P}}_e} \sum\limits_{r = 1}^{{N_e}} {\sum\limits_{s = 1}^{{N_e}}
 {{{\left\{ {{{{\bf P}_e^H}}{\bf{R}}_{E,R}^0{\bf{P}}_e} \right\}}_{rs}}} }
\end{equation}
\begin{equation*} \label{eq:pilot_attack_contr}
{\rm s.t.} \quad  \tr\left({{\bf P}_e} {\bf{W}}_e {\bf{W}}_e^H {{\bf P}_e^H}\right) \leq N_e \tau.
\end{equation*}

Let ${\bf{P}}_e = \left[ {{{\bf{p}}_1},{{\bf{p}}_2}, \cdots ,{{\bf{p}}_{{N_e}}}} \right]$,
where ${{\bf{p}}_s} \in \mathbb{C} {^{N_e \times 1}}$ is the $s$th column  of ${\bf{P}}_e$.
Then, we have the following theorem.
\begin{theorem}\label{theo:optimal_attack}
The optimal pilot contamination precoder ${\bf P}_e$ which solves (\ref{eq:pilot_attack}) has to satisfy the following condition
\begin{equation} \label{eq:optimal_precoder}
\sum\limits_{s = 1}^{{N_e}} {{{\bf{p}}_s}} = \sqrt {{N_e}} {{\bf{u}}_e}
\end{equation}
where ${{\bf{u}}_e} \in \mathbb{C} {^{N_e \times 1}}$ is the eigenvector corresponding to the largest eigenvalue of ${\bf{R}}_{E,R}^0$.
\end{theorem}
\begin{proof}
Please refer to Appendix \ref{proof:theo:optimal_attack}.
\end{proof}

Theorem \ref{theo:optimal_attack} indicates that transmitting the pilot sequence along the direction of the
eigenvector corresponding to the maximum eigenvalue of the receive correlation matrix of the eavesdropper's channel
constitutes the best possible attack strategy from the eavesdropper's
point of view.

}

\subsection{Downlink Data Transmission}
For the data transmission phase, we assume that the BSs in
all $L + 1$ cells perform jamming to prevent eavesdropping in their own cells.
Then, the transmit signal in the $l$th cell, $l = 0,1,\cdots,L$, is given by
 \begin{align}\label{eq:xn}
{{\bf{x}}_l} = \sqrt P \left( {\sqrt p \sum\limits_{k = 1}^K {{{\bf{w}}_{lk}}{s_{lk}}}  + \sqrt q {{\bf{U}}_{{\rm null},\,l}}{{\bf{z}}_l}} \right),
\end{align}
where $P$ is the average transmit power for downlink transmission,
$s_{lk}$ is the transmit signal for the $k$th user in the $l$th cell with
$E\left[|s_{lk}|^2\right] = 1$, and $p$ and $q$ denote the fractions of
power allocated to  transmit signal and AN, respectively.
To avoid the high implementation complexity associated with
the matrix inversion required for zero forcing  and MMSE precoding \cite{Gao2016JSAC}, in this paper, we adopt
simple matched filter precoding, as is typical for massive MIMO systems \cite{Marzetta2008TWC,Jose2011TWC,Zhu2014}.
Thus,
we set ${{\bf{w}}_{lk}} = \frac{{\widehat {\bf{h}}_{lk}^l}}{{\left\| {\widehat {\bf{h}}_{lk}^l} \right\|}}$ for the
precoding vector of the $k$th user in the $l$th cell.
Furthermore, in (\ref{eq:xn}), ${\bf{U}}_{{\rm null},\,l}$ and $\mathbf{z}_l \sim \mathcal{CN} \left( {\mathbf{0}, {{\bf{I}}_{{N_t}}}} \right)$
 denote the AN shaping matrix and the AN vector in the $l$th cell, respectively. We introduce $\widehat {\bf{H}}_l^l
 = \left[\mathbf{\widehat {\bf{h}}}_{l1}^l, \mathbf{\widehat {\bf{h}}}_{l2}^l,\cdots, \mathbf{\widehat {\bf{h}}}_{lK}^l\right]$ for notational simplicity.
{\bl
If the AN shaping matrix is chosen
as the NS of
$\widehat {\bf{H}}_l^l$ as is  conventionally done \cite{Goel2008TWC},
 for each new  channel estimate,  a matrix inversion is needed
 to compute the AN shaping matrix.
This leads to a high implementation complexity considering the
large numbers of antennas in massive MIMO. Thus, to
keep the implementation complexity low despite the large numbers of antennas,
we adopt for  the AN shaping matrix  the asymptotic  NS
of $\widehat {\bf{H}}_l^l$.
In particular, since based on \cite[Corollary 1]{Evans2000TIT}
 we have $ \frac{1}{N_t}{\left( {\widehat {\bf{H}}_l^l} \right)^H}\widehat {\bf{H}}_l^l \!\mathop
 \to  \limits^{{N_t} \to \infty } \! \!
\frac{1}{N_t}{\rm diag} \! {\left[{\tr\left( {\widehat {\bf{R}}_{l1}^l} \right)},  \! {\tr\left( {\widehat {\bf{R}}_{l2}^l} \right)},\!\cdots,\! {\tr\left( {\widehat {\bf{R}}_{lK}^l} \right) } \!\right]}$, we set ${{\bf{U}}_{{\rm null},\,l} } = {{\bf{I}}_{{N_t}}} - \widehat {\bf{H}}_l^l{\rm diag}\left[{\tr\left( {\widehat {\bf{R}}_{l1}^l} \right)^{-1}}, {\tr\left( {\widehat {\bf{R}}_{l2}^l} \right)^{-1}},\cdots, {\tr\left( {\widehat {\bf{R}}_{lK}^l} \right)^{-1} }\right]{\left( {\widehat {\bf{H}}_l^l} \right)^H}$.}
It can be shown that \cite[Corollary 1]{Evans2000TIT}
\begin{align}\label{eq:null_asy}
\frac{1}{N_t}\tr\left( {{\bf{U}}_{{\rm null},\,l} }  {{\bf{U}}^H_{{\rm null},\,l} }\right) \mathop  \to \limits^{{N_t} \to \infty } \frac{1}{N_t} (N_t - K).
\end{align}
Based on (\ref{eq:xn}), (\ref{eq:null_asy}), and the expression for ${{\bf{w}}_{lk}}$, we set $K p + (N_t - K)  q = 1$
to ensure $ \mathbf{x}_l^H  \mathbf{x}_l \mathop  \to \limits^{{N_t} \to \infty } P$.

The received signals at the $m$th user in the reference cell, ${y_{0m}}$, and at the eavesdropper, $\mathbf{y}_{\rm eve} \in {\mathbb{C}^{N_e  \times 1}} $,
are given by
 \begin{align}
{y_{0m}}  = \sum\limits_{l = 0}^L {{{\left( {{\bf{h}}_{0m}^l} \right)}^H}{{\bf{x}}_l}}  + {{n}}_{0m} \label{eq:y_0m}
\end{align} and
\begin{align}
{\mathbf{y}_{\rm eve}} = \sum\limits_{l = 0}^L {{{\left( {{\bf{H}}_{E}^l} \right)}^H}{{\bf{x}}_l}}  + {\mathbf{{n}}}_{\rm eve}, \label{eq:y_eve}
\end{align}
respectively.
Here, ${{n}}_{0m}$ and ${\mathbf{n}}_{\rm eve} \in {\mathbb{C}^{N_e  \times 1}} $ are zero-mean Gaussian noise processes with variance $N_{0,\rm{d}}$ and
covariance matrix $N_{0,\rm{d}} \mathbf{I}_{N_e}$, respectively.  We define the signal-to-noise ratio (SNR) for downlink data transmission as $\gamma = P/ N_{0,\rm{d}}$.
{\bl The structure of the downlink received signal  at the desired user and the eavesdropper
 is illustrated in  Figure \ref{downlink}. }

\begin{figure*}[!ht]
\centering
\includegraphics[width=0.9\textwidth]{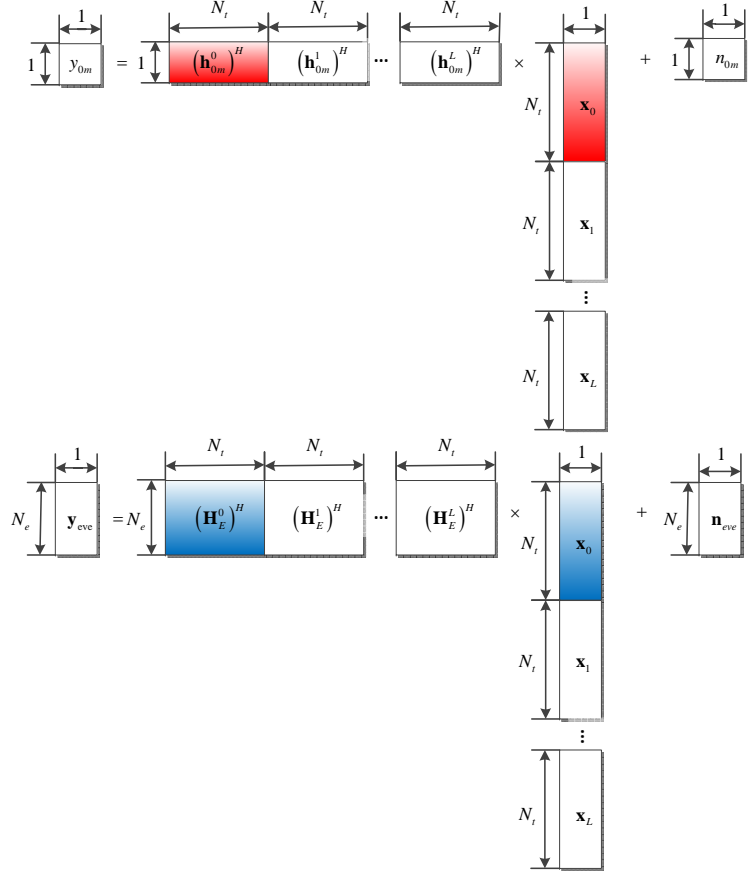}
\caption {\space\space Structure of the downlink received signal  for the desired user and the eavesdropper.}
\label{downlink}
\end{figure*}

In order to illustrate the secrecy threat
that the pilot contamination attack poses
for massive MIMO systems, we simulate the secrecy rate for MF precoding without AN generation\footnote{MF precoding without
AN generation is often adopted in massive MIMO systems without secrecy consideration\cite{Marzetta2008TWC,Jose2011TWC}.}(i.e., $p = 1, q = 0$)
 in the following example.

\begin{example}
Let $N_t = 128$, $N_e = 1$, $L = 3$, $K = 5$, $P_{lk} = 1$, $\forall k$, $\forall l$,  $P_E = 0.5$, and $N_0 = 1$.
{\bl The achievable ergodic secrecy rate \cite{Zhu2014} of a massive MIMO system
with MF precoding without AN generation} is shown in Table I
for i.i.d. fading  and different SNR values.

\begin{table*}[!t]\label{tab:secracy_rate_snr}
\captionstyle{center}
\caption{The achievable secrecy rate for different SNR values.}
\centering
\begin{tabular}{|c|c|c|c|c|c|c|}
\hline
  $\gamma$ (dB)  &  \  $-10$     \   &   \   $-8$     \   &  \  $-6$    \   &   \  $-4$    \   &  \ $-2$  \ &  \ $0$  \    \\ \hline
   \rm{Secrecy Rate}  (b/s/Hz) &   \   $0.3828$   \   &   \   $0.4097$   \   &   \   $0.3807$   \   &   \   $0.2759$   \ & \  $0.0822$  \ & \ $0$  \  \\
\hline
\end{tabular}
\vspace*{10pt}

\hrulefill
\end{table*}

\end{example}

{\bl The results in Table I reveal that
even if the pilot power of the eavesdropper is only half of the pilot power of the desired user
and the eavesdropper has only a single antenna, zero secrecy rate may result.}

\section{Massive MIMO Signal Design for Combating the Pilot Contamination Attack}
In this section, we investigate three signal designs for combating the pilot
contamination attack by an active eavesdropper in a multi-cell multi-user massive MIMO system.

\subsection{MF-AN Design}
An achievable ergodic secrecy rate of the massive MIMO system described
in Section II is given by\cite{Zhu2014}
 \begin{align}\label{eq:R_sec}
{R_{\sec}} = {\left[ {{R_{0m}} - {C_{\rm{eve}}} } \right]^ + },
 \end{align}
 where $R_{0m}$ and $C_{\rm{eve}}$ denote an achievable ergodic rate between the BS and the $m$th user
and  the ergodic capacity between the BS and the eavesdropper in the $0$th cell, respectively.  The
achievable ergodic rate ${R_{0m}}$ is given by \cite[Eq. (8)]{Zhu2014}
 \begin{align}\label{eq:R_B}
{R_{0m}} = E\left[{\log _2}\left( {1 + {\rm SINR}_{0m}} \right) \right]
 \end{align}
 where ${\rm SINR}_{0m}$ is given by
\begin{align}\label{eq:SINR_0m}
{\rm SINR}_{0m} = \frac{p \gamma{{ { { \left| {{{\left( {{\bf{h}}_{0m}^0} \right)}^H}{{\bf{w}}_{0m}}}  \right|}}^2}}}{A}
\end{align}
\begin{multline}\label{eq:SINR_0m_b}
A   =  p \gamma \sum\limits_{k = 1,k \ne m}^K { {{{\left| {{{\left( {{\bf{h}}_{0m}^0} \right)}^H}{{\bf{w}}_{0k}}} \right|}^2}}}  +  q \gamma  {{{\left| {{{\left( {{\bf{h}}_{0m}^0} \right)}^H}{{\bf{U}}_{{\rm null},\,0}}} \right|}^2}}  \\
+  p\gamma   \sum\limits_{l = 1}^L  {\sum\limits_{k = 1}^K { {{{\left| {{{\left( {{\bf{h}}_{0m}^l} \right)}^H}{{\bf{w}}_{lk}}} \right|}^2}} } }
   +q \gamma \sum\limits_{l = 1}^L { {{{\left| {{{\left( {{\bf{h}}_{0m}^l} \right)}^H}{{\bf{U}}_{{\rm null},\,l}}} \right|}^2}} }    +  1.
\end{multline}

{\bl
In this paper, we  make the pessimistic assumption  that
the eavesdropper has perfect knowledge of its own channel and is able to decode and
cancel the signals of all intra-cell and inter-cell users from
the received signal ${\mathbf{y}_{\rm eve}}$ in (\ref{eq:y_eve}) except for the signal intended for the $m$th user in the $0$th cell.
This assumption results in an upper bound on the eavesdropper's capacity, and consequently, in a
lower bound on the ergodic secrecy rate.
If the eavesdropper has access to the data of all
intra-cell and inter-cell interfering users, this low bound
is achievable. This might be the case if the interfering users
cooperate with the eavesdropper. We note that this assumption
constitutes a worst-case scenario. Hence, if secure communication
can be achieved  for this worst case, then secure communication
can also be achieved for more optimistic settings (e.g. when the eavesdropper cannot mitigate
all multi-user interference).}
Considering the worst case, ${C_{\rm{eve}}}$
can be expressed as \cite[Eq. (7)]{Zhu2014}
\begin{align}
& {C_{\rm eve}} \nonumber \\
&= E\left[ {{{\log }_2}\left( {1 + p\gamma {{\left( {{{\bf{w}}_{0m}}} \right)}^H}{\bf{H}}_E^0{{\bf{Q}}^{ - 1}}{{\left( {{\bf{H}}_E^0} \right)}^H}{{\bf{w}}_{0m}}} \right)} \right] \nonumber \\
 & = E\left[ {{{\log }_2}\left( {1 + \frac{{p\gamma }}{{{{\left\| {\widehat {\bf{h}}_{0m}^0} \right\|}^2}}}{{\left( {\widehat {\bf{h}}_{0m}^0} \right)}^H}
 {\bf{H}}_E^0{{\bf{Q}}^{ - 1}}{{\left( {{\bf{H}}_E^0} \right)}^H}\widehat {\bf{h}}_{0m}^0} \right)} \right], \label{eq:C_eve_1}
\end{align}
where
\begin{align} \label{eq:Q}
{\bf{Q}} = q\gamma \sum\limits_{l = 0}^L {{{\left( {{\bf{H}}_E^l} \right)}^H}{{\bf{U}}_{{\rm null},\,l}}{{ {{{\bf{U}}^H_{{\rm null},\,l}}}}}{\bf{H}}_E^l}  + {{\bf{I}}_{{N_e}}}
\end{align}
denotes the noise
correlation matrix at the eavesdropper.

In the following theorem, we provide an asymptotic achievable secrecy rate expression
when the number of transmit antennas $N_t$ tends to infinity. {\bl
For convenience, the notation used in the theorem is summarized in
Table \ref{tab:theorem 2}.

\begin{table*}[!t]
\captionstyle{center}
\centering
{\bl
\caption{Notation used in Theorem \ref{prop:sec_rate_mul}.}
\label{tab:theorem 2}
\begin{tabular}{|c|c|}
\hline
 Notation   &           Description       \\ \hline
 ${\rm SINR}_{0m,\,{\rm asy}}$       &  \quad  Output SINR of the desired user in the asymptotic regime $N_t \rightarrow\infty$  \quad  \\
\hline
 ${\rm SINR}_{\rm eve,\, asy}$    &  \quad  Output SINR of the eavesdropper in the asymptotic regime $N_t \rightarrow\infty$ \quad  \\
\hline
   ${\theta _m}$  &  \quad  Power of signal intended for  the $m$th user in the $0$th cell received by the desired user \quad  \\
\hline
   ${\theta _{b,p}}$  &  \quad   Power of multi-user interference at the desired user  \quad  \\
\hline
   ${\theta _{b,q}}$  &  \quad   Power of the AN at the desired user   \quad  \\
\hline
  $ \Lambda _{0m}^l$   &  \quad  Power of the signal intended for the $m$th user in the $l$th cell  received at the desired user   \quad  \\
\hline
  $\mathbf{{Q}}_{\rm{asy}}$   &  \quad  Power of  AN received at the eavesdropper    \quad  \\
\hline
  $\eta _{ij}^0$   &  \quad   Power  of the signal intended for the $m$th user in the $0$th cell received at the eavesdropper    \quad  \\
\hline
  $\eta _{ij}^l$   &  \quad  Power of the signal intended for the $m$th user in the $l$th cell received at the eavesdropper    \quad  \\
\hline
\end{tabular} }
\vspace*{10pt}

\hrulefill
\end{table*}

}

\begin{theorem}\label{prop:sec_rate_mul}
An asymptotic achievable secrecy rate for a multi-cell multi-user massive MIMO system
employing the MF-AN design to overcome a multi-antenna active eavesdropper
is given by
 \begin{multline}\label{eq:asy_sec}
{R_{\sec, \, \rm{asy}}}  \mathop  \rightarrow \limits^{{N_t} \to \infty } \left[  {\log _2}\left( {1 + {\rm SINR}_{0m,\,{\rm asy}} } \right) \right. \\
\left.- {\log _2}\left( {1 + {\rm SINR}_{{\rm eve,\,asy}}} \right) \right]^{+},
\end{multline}
where
 \begin{align}
{\rm SINR}_{0m,\,{\rm asy}} &= \frac{{ p \gamma {\theta _m}}}{{p\gamma {\theta _{b,p}} +  q\gamma {\theta _{b,q}} + 1}},  \label{eq:b} \\
{\rm SINR}_{\rm eve,\, asy} &= \frac{p\gamma {\theta _{e}}}{\tr \left({\widehat {\bf{R}}_{0m}^0} \right) },   \label{eq:e}
\end{align}
with
 \begin{multline}\label{eq:theta_m}
 {\theta _m}  = \tr\left( {\widehat {\bf{R}}_{0m}^0} \right) + \tr\left( {\widehat {\bf{R}}_{0m}^0} \right)^{-1} {{\tr\left({\left({{\bf{R}}_{0m}^0 - \widehat{\bf{R}}_{0m}^0} \right)
  \widehat{\bf{R}}_{0m}^0} \right)}}
  \end{multline}
   \begin{multline}\label{eq:theta_bp}
{\theta _{b,p}}  =  \sum\limits_{l = 0}^L {\sum\limits_{k = 1,k \ne m}^K {\tr{{\left( {\widehat {\bf{R}}_{lk}^l} \right)}^{ - 1}}
\tr\left( {{\bf{R}}_{0m}^l\widehat {\bf{R}}_{lk}^l} \right)}}  \\
  +  \sum\limits_{l = 1}^L { \tr{{\left(\! {\widehat {\bf{R}}_{lm}^l} \right)}^{ - 1}}  \Lambda _{0m}^l}
  \end{multline}
    \begin{multline}\label{eq:theta_bq}
 {\theta _{b,q}}  = \sum\limits_{l = 0}^L {\tr\left( {{\bf{R}}_{0m}^l} \right)}  - \sum\limits_{l = 0}^L \sum\limits_{k = 1,k \ne m}^K {\tr{{\left( {\widehat {\bf{R}}_{lk}^l} \right)}^{ - 1}}\tr\left( {{\bf{R}}_{0m}^l\widehat {\bf{R}}_{lk}^l} \right)}   \\
  - \sum\limits_{l = 0}^L {\tr{{\left( {\widehat {\bf{R}}_{lm}^l} \right)}^{ - 1}}\Lambda _{0m}^l}.
 \end{multline}
 Here, $\Lambda _{0m}^l$ in (\ref{eq:theta_bp}) and (\ref{eq:theta_bq}) is given by
 \begin{multline}
   \Lambda _{0m}^l   = {\tau ^2}{P_{0m}}{\left| {\tr\left( {{\bf{C}}_{lm}^l{\bf{R}}_{0m}^l} \right)} \right|^2}  + \tau {N_0}\tr\left( {{\bf{R}}_{0m}^l{\bf{C}}_{lm}^l{{\left( {{\bf{C}}_{lm}^l} \right)}^H}} \right) \\  + {\tau ^2}\sum\limits_{t = 1}^L {{P_{tm}}\tr\left( {{\bf{R}}_{0m}^l{\bf{C}}_{lm}^l{\bf{R}}_{tm}^l{{\left( {{\bf{C}}_{lm}^l} \right)}^H}} \right)}
    \\  + {\tau ^2}{P_E} {r_{E,R}^0} \tr\left( {{\bf{R}}_{0m}^l{\bf{C}}_{lm}^l{\bf{R}}_{E,T}^l{{\left( {{\bf{C}}_{lm}^l} \right)}^H}} \right) \\
    \end{multline}
 \begin{multline}\label{eq:Clm}
{\bf{C}}_{lm}^l  =  \sqrt {{P_{lm}}} {\bf{R}}_{lm}^l \\
\times {\left( {{N_0}{{\bf{I}}_{{N_t} }}  + \tau \left( {\sum\limits_{t = 0}^L {{P_{tm}}{\bf{R}}_{tm}^l}   +  {P_E} {r_{E,R}^l{\bf{R}}_{E,T}^l} } \right)} \right)^{ - 1}} .
  \end{multline}
Furthermore, ${\theta _e}$ in (\ref{eq:e}) is given by
\begin{align}
{\theta _e}  =  { \sum\limits_{i = 1}^{{N_e}} {\sum\limits_{j = 1}^{{N_e}} {{{\left\{ {{\bf{Q}}_{\rm{asy}}^{ - 1}} \right\}}_{ij}}\eta _{ij}^0} } },
\end{align}
where
 \begin{multline} \label{eq:eta_prop}
\eta _{ij}^l  = {\tau ^2}{\left\{ {{\bf{R}}_{E,R}^l} \right\}_{ij}}\sum\limits_{t = 0}^L {{P_{tm}}\tr\left( {{\bf{R}}_{E,T}^l{\bf{C}}_{lm}^l{\bf{R}}_{tm}^l{{\left( {{\bf{C}}_{lm}^l} \right)}^H}} \right)} \\
 + {\tau ^2}\frac{{{P_E}}}{{{N_e}}}\sum\limits_{r = 1}^{{N_e}} {{{\left\{ {{\bf{R}}_{E,R}^l} \right\}}_{ir}}} \sum\limits_{r = 1}^{{N_e}} {{{\left\{ {{\bf{R}}_{E,R}^l} \right\}}_{rj}}} {\left| {\tr\left( {{\bf{C}}_{lm}^l{\bf{R}}_{E,T}^l} \right)} \right|^2}  \\
  + {N_0}\tau {\left\{ {{\bf{R}}_{E,R}^l} \right\}_{ij}}\tr\left( {{\bf{R}}_{E,T}^l{\bf{C}}_{lm}^l{{\left( {{\bf{C}}_{lm}^l} \right)}^H}} \right)
\end{multline}
and $\mathbf{{Q}}_{\rm{asy}} = q \gamma \sum\nolimits_{l = 0}^{{L}} \mathbf{Q}_l + \mathbf{I}_{N_r}$. Matrix $\mathbf{Q}_l \in \mathbb{C}{^{{N_e} \times N_e}} $
has elements
\begin{multline}
{\left\{ {{{\bf{Q}}_l}} \right\}_{ij}} =  {\left\{ {{\bf{R}}_{E,R}^l} \right\}_{ij}}\tr\left( {{\bf{R}}_{E,T}^l} \right) - \left[ {\left\{ {{\bf{R}}_{E,R}^l} \right\}_{ij}} \right. \\
\left. \times \sum\limits_{k = 1,k \ne m}^K {\tr{{\left( {\widehat {\bf{R}}_{lk}^l} \right)}^{ - 1}}} \tr\left( {{\bf{R}}_{E,T}^l\widehat {\bf{R}}_{lk}^l} \right) + \tr{\left( {\widehat {\bf{R}}_{lm}^l} \right)^{ - 1}}\eta _{ij}^l \right].
\end{multline}
\begin{proof}
Please refer to Appendix \ref{proof:prop:sec_rate_mul}.
\end{proof}
\end{theorem}

Theorem \ref{prop:sec_rate_mul} provides an expression for
the general asymptotic achievable secrecy rate for ${{N_t} \to \infty }$, which
is valid for arbitrary $p$ and $q$.  In practice, the optimal $p$ can be  found by performing a simple one
dimensional numerical search in the interval $0 \leq p \leq \frac{1}{K}$ for maximization of (\ref{eq:asy_sec}).
For the special case of a single-antenna eavesdropper, the optimal $p$ can be obtained in closed form,
c.f. Section \ref{sec:single-antenna}-A.

{\bl For the conventional secrecy problem existing in
multiple-input, single-output systems with multiple-antenna
passive eavesdroppers \cite{Khisti2010TIT}, the optimal transmission design does
not require AN generation and the secrecy rate increases for increasing SNR.
To illustrate  the importance of AN generation when an active eavesdropper is present, we introduce
the following theorem.
\begin{theorem} \label{prop:all_transmit}
For the case without AN generation where $p = \frac{1}{K}$, ${R_{\sec ,\, {\rm asy}}}$ in (\ref{eq:asy_sec}) is a monotonically decreasing function of $\gamma$
for $\gamma > \gamma_{\rm th}$, where
\begin{align}\label{eq:gamma_th}
{\gamma _{{\rm{th}}}} & = \frac{{ - {\theta _{b,p}}{\tilde{\theta} _{e}} + \sqrt {\theta _{b,p}^2\tilde{\theta} _{e}^2 + \left( {{\theta _{b,p}} + {\theta _m}} \right)\left( {{\theta _m} - {\tilde{\theta}_{e}}} \right){\theta _{b,p}}{\tilde{\theta} _{e}}} }}{{\left( {{\theta _{b,p}} + {\theta _m}} \right){\theta _{b,p}}{\tilde{\theta} _{e}}}}
\end{align}
with
\begin{align}\label{eq:gamma_th_2}
{\tilde{\theta} _{e}} & = \frac{\sum\limits_{i = 1}^{{N_e}} \eta _{ii}^0} {\tr \left({\widehat {\bf{R}}_{0m}^0} \right) } .
\end{align}
\begin{proof}
Please refer to Appendix \ref{proof:prop:all_transmit}.
\end{proof}
\end{theorem}

{\emph{Remark 3:}} Theorem \ref{prop:all_transmit} indicates that in the presence of an active eavesdropper,
the secrecy rate will decrease
for increasing SNR in the high SNR regime if all the available  power at the BS is allocated to the information-carrying
signals.
This behaviour is very different from that of multiple-antenna
systems with passive eavesdroppers \cite{Khisti2010TIT}, where
the secrecy rate is a monotonically increasing function of the SNR throughout the entire
SNR region even if AN is not employed.
However, for the case considered in this paper,
the BS relies on estimated CSI for  precoder design,
which is contaminated by the attack of the active eavesdropper.
Since the BS cannot distinguish the actual channel of
the desired user from the eavesdropper channel, the precoder may implicitly beamform the information signal towards
the eavesdropper.  As a result, the capacity of the eavesdropper may increase
faster than the achievable rate of the desired user as the transmit power increases.
Hence, it is advantageous if  a fraction of power is used to generate AN to degrade the eavesdropper's
ability to decode the transmit signal intended for the desired user.  This is the motivation behind
the MF-AN design. }

\subsection{NS Design} \label{sec:NS}
In order to obtain some insight for transmit signal design, in the following theorem,
we investigate the
asymptotic achievable secrecy rate under the assumption that the transmit correlation matrices
of the users are orthogonal to the transmit correlation matrices of the eavesdropper.
{\bl The notations used in the following theorem are summarized in Table \ref{tab:theorem 3},
where the subscript ``orth" in the variable names indicates that the transmit correlation matrices
of the users are assumed to be orthogonal to the transmit correlation matrices of the eavesdropper.}

\begin{table*}[!t]
\captionstyle{center}
\centering
{\bl
\caption{Notation used in Theorem \ref{prop:sec_orth}.}
\label{tab:theorem 3}
\begin{tabular}{|c|c|}
\hline
 Notation   &           Description      \\ \hline
 ${\rm SINR}_{0m,\,{\rm asy},\,\rm{orth}} $       &  \quad  Output SINR of the desired user in the asymptotic regime
 $N_t \rightarrow\infty$   \quad  \\
\hline
   ${\theta _{m,\rm{orth}}}$  &  \quad  Power of signal intended for the $m$th user in the $0$th cell received by the desired user    \quad  \\
\hline
   ${\theta _{b,p,\rm{orth}}}$  &  \quad   Power of  multi-user interference at the desired user   \quad  \\
\hline
$  {\theta _{b,q,\rm{orth}}}$  &  \quad  Power of the AN received at the desired user  \quad  \\
\hline
  $ \Lambda _{0m,\rm{orth}}^l $   &  \quad  Power  of signal intended for the $m$th user in the  $l$th cell  received at the desired user \quad  \\
\hline

\end{tabular} }
\vspace*{10pt}

\hrulefill
\end{table*}

\begin{theorem}\label{prop:sec_orth}
If $\sum\nolimits_{t = 0}^L  \tr\left( {{\bf{R}}_{tm}^l}{{\bf{R}}_{E,T}^{l}}\right) = 0$ for $l = 0,1,\cdots,L$,
the secrecy rate ${R_{\sec, \, \rm{asy}}}$ is identical to the achievable rate
of the desired user $R_{\sec, \, \rm{asy}, \, \rm{orth}} =
\log_2(1 + {\rm SINR}_{0m,\,\rm{asy},\,\rm{orth}})$,
where
 \begin{align}\label{eq:0m_orth}
{\rm SINR}_{0m,\,{\rm asy},\,\rm{orth}} = \frac{{ p \gamma {\theta _{m,\rm{orth}}}}}{{p\gamma {\theta _{b,p,\rm{orth}}} +  q\gamma {\theta _{b,q,\rm{orth}}} + 1}}
\end{align}
with
 \begin{multline}
{\theta _{m,\rm{orth}}}  = \tr\left( {\widehat {\bf{R}}_{0m,\rm{orth}}^0} \right) \\
+ \tr{\left( {\widehat {\bf{R}}_{0m,\rm{orth}}^0} \right)^{ - 1}}
\tr\left( {\left({{\bf{R}}_{0m}^0 - \widehat {\bf{R}}_{0m,\rm{orth}}^0} \right)\widehat {\bf{R}}_{0m,\rm{orth}}^0} \right) \\
\end{multline}
 \begin{multline}
{\theta _{b,p,\rm{orth}}}  =   \sum\limits_{l = 0}^L {\sum\limits_{k = 1,k \ne m}^K {\tr{{\left( {\widehat {\bf{R}}_{lk}^l} \right)
}^{ - 1}}\tr\left( {{\bf{R}}_{0m}^l\widehat {\bf{R}}_{lk}^l} \right)}}  \\
+  \sum\limits_{l = 1}^L { \tr{{\left({\widehat {\bf{R}}_{lm,{\rm orth}}^l} \right)}^{ - 1}}  \Lambda _{0m,{\rm orth}}^l}
\end{multline}
 \begin{multline}
{\theta _{b,q,\rm{orth}}}   =  \sum\limits_{l = 0}^L {\tr\left( {{\bf{R}}_{0m}^l} \right)}  \\ -  \sum\limits_{l = 0}^L \sum\limits_{k = 1,k \ne m}^K {\tr{{\left( {\widehat {\bf{R}}_{lk}^l} \right)}^{ - 1}}\tr\left( {{\bf{R}}_{0m}^l\widehat {\bf{R}}_{lk}^l} \right)} \\
 - \sum\limits_{l = 0}^L {\tr{{\left( {\widehat {\bf{R}}_{lm,\rm{orth}}^l} \right)}^{ - 1}}\Lambda _{0m,\rm{orth}}^l }
\end{multline}
 \begin{multline}
 \Lambda _{0m,\rm{orth}}^l  = {\tau ^2}{P_{0m}}{\left| {\tr\left( {{\bf{C}}_{lm,\rm{orth}}^l{\bf{R}}_{0m}^l} \right)} \right|^2} \\
  + {\tau ^2}\sum\limits_{t = 1}^L {{P_{tm}}\tr\left( {{\bf{R}}_{0m}^l{\bf{C}}_{lm,\rm{orth}}^l{\bf{R}}_{tm}^l {{\left( {{\bf{C}}_{lm,\rm{orth}}^l} \right)}^H}} \right)}   \\
  +   \tau {N_0}\tr\left( {{\bf{R}}_{0m}^l{\bf{C}}_{lm,\rm{orth}}^l{{\left( {{\bf{C}}_{lm,\rm{orth}}^l} \right)}^H}} \right)
\end{multline}
\begin{align}\label{eq:Cllm}
 {\bf{C}}_{lm,\rm{orth}}^l     = \sqrt {{P_{lm}}} {\bf{R}}_{lm}^l{\left( {{N_0}{{\bf{I}}_{{N_t} }}
  +  \tau {\sum\limits_{t = 0}^L {{P_{tm}}{\bf{R}}_{tm}^l} } } \right)^{ - 1}}
\end{align}
\begin{align}
 \widehat {\bf{R}}_{lm,\rm{orth}}^l    = {P_{lm}}\tau {\bf{R}}_{lm}^l{\left( {{N_0}{{\bf{I}}_{{N_t}}} + \tau \sum\limits_{t = 0}^L {{P_{tm}}{\bf{R}}_{tm}^l}  } \right)^{ - 1}}{\bf{R}}_{lm}^l.
\end{align}
\begin{proof}
Please refer to Appendix \ref{proof:prop:sec_orth}.
\end{proof}
\end{theorem}

Theorem \ref{prop:sec_orth} reveals that when
the channels of the eavesdropper  and the users are statistically orthogonal, the pilot contamination attack has no impact on the secrecy rate.
It is known that for typical massive MIMO scenarios, the transmit correlation matrices of the channels are
low rank \cite{Yin2013JSAC,Adhikary2013TIT,Yin2014JSTSP,Nam2014JSTSP,Adhikary2014JSAC,Sun2015TCOM,Shen2015CL}. Hence, inspired by Theorem \ref{prop:sec_orth},
in the remainder of this subsection, we introduce an NS based design where the information-carrying signal is transmitted in
the NS of the transmit correlation matrix of the eavesdropper's channel $\mathbf{H}_{E}^{l}$.

Assume the rank of $\mathbf{R}_{E,T}^l$ is $T_l$, $l = 0,1,\cdots,L$.
Let us construct a matrix ${{\bf{V}}_E^l} \in {\mathbb{C}^{{N_t} \times M_l}}$,
whose $M_l$ columns are the $M_l$ eigenvectors corresponding to the zero eigenvalues of $\mathbf{R}_{E,T}^l$, where $M_l = N_t - T_l$.
In the uplink training phase, we multiply ${\bf{Y}}_{l}$  with ${{\bf{V}}_E^l}$ to obtain\footnote{The $m$th user in the $l$th cell ($l = 1,2,\cdots,L$) is also affected by the pilot sequence
sent by the active eavesdropper in the reference cell. As a result, we assume that the $l$th BS also employs the NS design to combat the pilot contamination
attack.}
 \begin{multline} \label{eq:y_0_multi_null}
 \left({{\bf{V}}_E^l}\right)^H  {{\bf{Y}}_l}  =   \sum\limits_{k = 1}^K  \! {\sqrt {{P_{lk}}}   \left({{\bf{V}}_E^l}\right)^H{\bf{h}}_{lk}^l} {\boldsymbol{\omega }}^T_{k} \\ + \sum\limits_{l = 1}^L  \sum\limits_{k = 1}^K  \sqrt {{P_{lk}}} \left({{\bf{V}}_E^l}\right)^H{\bf{h}}_{lk}^l {\boldsymbol{\omega }}_{k}^T  \\ + \sqrt {{P_E}} \left({{\bf{V}}_E^l}\right)^H
{\bf{{H}}}_{E}^l {\bl \mathbf P_e} {\bf{W}}_e +  \left({{\bf{V}}_E^l}\right)^H {\bf{N}}
\end{multline}
and
 \begin{multline}\label{eq:y_uplink_null}
 {{\bf{y}}_{l,\rm null}}  = {\rm vec}\left( { \left({{\bf{V}}_E^l}\right)^H  {\bf{Y}}_l} \right) = \sum\limits_{k = 1}^K   {\sqrt {{P_{lk}}} \left( {{{\boldsymbol{\omega }}_{k}} \otimes {{\bf{I}}_{{N_t}}}} \right){\bf{h}}_{lk,\rm{null}}^l}  \\
  +  \sum\limits_{l = 1}^L  {\sum\limits_{k = 1}^K   {\sqrt {{P_{lk}}} \left( {{{\boldsymbol{\omega }}_{k}} \otimes {{\bf{I}}_{{N_t}}}} \right){\bf{h}}_{lk,\rm{null}}^l} }  \\
  + \sqrt {\frac{{{P_E}}}{{{N_e}}}}  \left( {{{\boldsymbol{\omega }}_{m}} \otimes {{\bf{I}}_{{N_t}}}} \right) \sum\limits_{r = 1}^{{N_e}} {\bl {{\bf{h}}_{{\rm eff},r,\rm{null}}^l}} + {\bf{\widetilde{n}}}_{\rm null}^{l},
\end{multline}
where
${\bf{h}}_{lk,\rm{null}}^l = \left({{\bf{V}}_E^l}\right)^H {\bf{h}}_{lk}^l  \sim \mathcal{CN}\left( { \mathbf{0},{{\bf{R}}_{lk,\rm{null}}^l}} \right)$,
 ${{\bf{R}}_{lk,\rm{null}}^l} = \left({{\bf{V}}_E^l}\right)^H {{\bf{R}}_{lk}^l} {{\bf{V}}_E^l} \in \mathbb{C}{^{{M}_l \times {M}_l }} $,
$l = 0,1,\cdots,L$, $k = 1,2,\cdots,K$, and $\widetilde {\bf{n}}_{\rm null}^{l} \sim \mathcal{CN}\left( {\mathbf{0}, \tau N_0 \mathbf{I}_{M_l}} \right)$.
${\bl {{\bf{h}}_{{\rm eff},r,\rm{null}}^l}}$ is the $r$th column of  matrix  $\left({{\bf{V}}_E^l}\right)^H
{\bf{{H}}}_{E}^l {\bl \mathbf{P}_e}$.
We obtain an estimate for $\mathbf{h}_{lm,\rm{null}}^l$ as
\begin{multline} \label{eq:mmse_h0m_null}
{\widehat {\bf{h}}_{lm,\rm{null}}^l} =
\sqrt {{P_{lm}}} {\bf{R}}_{lm,\rm{null}}^l  \\  \times {\left({{N_0}{{\bf{I}}_{{M}}}  +  \tau  {\sum\limits_{t = 0}^L
{{P_{tm}}{\bf{R}}_{tm,\rm{null}}^l} }} \right)^{ - 1}} {\widetilde {\bf{y}}_{lm,\rm{null}}},
\end{multline}
where
\begin{multline} \label{eq:y_0m_null}
{\widetilde {\bf{y}}_{lm,\rm{null}}} = \sqrt {{P_{lm}}} \tau {\bf{h}}_{lm,\rm{null}}^l + \sum\limits_{t = 1}^L {\sqrt {{P_{tm}}} \tau {\bf{h}}_{tm,\rm{null}}^l}
\\ + \sqrt {\frac{{{P_E}}}{{{N_e}}}} \sum\limits_{r = 1}^{{N_e}} {\bl {{\bf{h}}_{{\rm eff},r,\rm{null}}^l}}  + {\left( {{{\boldsymbol{\omega }}_m} \otimes {{\bf{I}}_{{M}}}} \right)^H}{\bf{\widetilde{n}}_{\rm null}}^l.
\end{multline}

The actual channel vector can be expressed as ${\bf{h}}_{lm,\rm{null}}^l = {\bf{\widehat h}}^l_{lm,\rm{null}} + {\bf{e}}_{lm,\rm{null}}^l$, where
${\bf{\widehat h}}^l_{lm} \sim \mathcal{CN} \left({\mathbf{0}, {\widehat {\bf{R}}_{lm,\rm{null}}^l} }\right)$ and ${\bf{e}}^l_{lm,\rm{null}} \sim \mathcal{CN} \left({\mathbf{0}, }{{{\bf{R}}_{lm,\rm{null}}^l} - {\widehat {\bf{R}}_{lm,\rm{null}}^l} }\right)$ with
 \begin{multline} \label{eq:R0m0_est_null}
\widehat {\bf{R}}_{lm,\rm{null}}^l = {P_{lm}}\tau {\bf{R}}_{lm,\rm{null}}^l \\
\times
{\left( {N_0}{{\bf{I}}_{{M}}}  +  \tau  {\sum\limits_{t = 0}^L {{P_{tm}}{\bf{R}}_{tm,\rm{null}}^l}} \right)^{ - 1}}
{\bf{R}}_{lm,\rm{null}}^l.
\end{multline}

In the downlink data transmission phase, we employ the same transmission scheme as in Section II-B, but replace $\mathbf{\bf{w}}_{lm}$
with ${{\bf{w}}_{lm,\rm{null}}} =  \mathbf{V}_E^l \frac{{\widehat {\bf{h}}_{lm,\rm{null}}^l}}{{\left\| {\widehat {\bf{h}}_{lm,\rm{null}}^l} \right\|}}$,
and set $p = 1/K$, $q = 0$.

{\emph{Remark 4:}} The proposed NS design transmits the signal
in the orthogonal subspace $\mathbf{V}_E^{0}$.  As a result, the
performance of the NS design depends on the rank of $\mathbf{V}_E^{0}$.
For instance, for the extreme case of i.i.d. fading,  $\mathbf{V}_E^{0}$ does not exist
and hence  the NS design is not applicable. In practice,  the NS design will be beneficial in highly correlated channels, for a strong pilot contamination attack, and in the high SNR
regime since it can effectively degrade the eavesdropper's performance in these scenarios.
In contrast, we expect the MF-AN design
to outperform the  NS design in weakly correlated channels, for a weak pilot contamination attack,
and in the low SNR regime.

\subsection{Unified Design}
Considering Remark 4, we propose a unified design that exploits the advantages of both the conventional MF-AN design
and the NS design. The corresponding transmit signal in the $l$th cell, $l = 0,1,\cdots,L$, is given by
\begin{multline}\label{eq:unified}
{{\bf{x}}_l} = \sqrt P  \left({ \sqrt{\alpha} \left(\sqrt {p} \sum\limits_{k = 1}^K {{{\bf{w}}_{lk}}{s_{lk}}}  + \sqrt { q} {{\bf{U}}_{\rm{null},l}}{{\bf{z}}_l} \right)  } \right.\\
 + \left.{\sqrt {\frac{\beta }{K}} \left( {{{\bf{w}}_{lm,\rm{null}}}{s_{lk}} + \sum\limits_{k = 1,k \ne m}^K {{{\bf{w}}_{lk}}{s_{lk}}} } \right)} \right),
\end{multline}
where $\alpha$ and $\beta$ denote the weights of
the MF-AN design and the NS design with $\alpha + \beta = 1$.
Using the secrecy rate
in (\ref{eq:R_sec}) as the cost function, the optimal $\alpha$ and $\beta$ can be obtained from a
one-dimensional numerical search.

\section{The Single-Antenna Eavesdropper Case} \label{sec:single-antenna}
In this section, we analyze the performance of massive MIMO systems
for the case of an active single-antenna eavesdropper \cite{Zhou2012TWC,Im2013,Kapetanovic2013}.
For the MF-AN design, we derive closed-form expressions for the optimal power allocation policy
for the transmit signal and the AN as well as the minimum transmit signal power required
to ensure secure transmission.
Also, we derive a closed-form expression for a threshold that can be used to optimally switch
between the MF-AN and the NS design.  Finally, we investigate the single-cell single-user
massive MIMO system.

\subsection{Multi-cell Multi-user Case}
{\bl We obtain a closed-form solution for the optimal power allocation for the MF-AN design
in the following theorem.  The notation used in the theorem is summarized in Table \ref{tab:theorem 4}}.
\begin{table*}[!t]

\captionstyle{center}
\centering
{\bl
\caption{Notations used in Theorem \ref{prop:opt_power}.}
\label{tab:theorem 4}
\begin{tabular}{|c|c|}
\hline
 Notation   &           Description       \\ \hline
 $a_1$, $b_1$, $c_1$      &   \quad  ${a_1}{p^2} + {b_1}p + {c_1}$ is the numerator of the ratio between  $1 + {\rm SINR}_{0m,\,{\rm asy}}$ and $1 + {\rm SINR}_{\rm eve,\, asy}$  when $N_e = 1$ \quad  \\
\hline
  $a_2$, $b_2$, $c_2$    &  \quad ${a_2}{p^2} + {b_2}p + {c_2}$ is the denominator of the ratio between  $1 + {\rm SINR}_{0m,\,{\rm asy}}$ and $1 + {\rm SINR}_{\rm eve,\, asy}$  when $N_e = 1$   \quad  \\
\hline
   ${\theta _{e,e}}$  &  \quad  Power of the signal intended for the $m$th user in the $0$th cell received at the eavesdropper  when $N_e = 1$  \quad  \\
\hline
   ${\theta _{e,q}} $  &  \quad  Power of AN received at the eavesdropper  when $N_e = 1$  \quad  \\
\hline
   $\Lambda _E^l $  &  \quad  Power of the signal intended for the $m$th user in the $l$th cell received at the eavesdropper  when $N_e = 1$   \quad  \\
\hline
\end{tabular} }
\vspace*{10pt}

\hrulefill
\end{table*}

\begin{theorem} \label{prop:opt_power}
Let us define
\begin{multline}\label{eq:p1}
{p_1} = -\frac{{\left( {{a_1}{c_2}   -   {a_2}{c_1}} \right)}}  { {{a_1}{b_2} - {a_2}{b_1}}} \\
- \frac{{ \sqrt {{{\left( {{a_1}{c_2}  -  {a_2}{c_1}} \right)}^2}   -  \left(  {{a_1}{b_2} - {a_2}{b_1}} \right) \left( {{b_1} - {b_2}} \right){c_1}} }} { {{a_1}{b_2} - {a_2}{b_1}}}
 \end{multline}
 and
\begin{multline}\label{eq:p2}
{p_2} = -\frac{{\left( {{a_1}{c_2}   -   {a_2}{c_1}} \right)}}  { {{a_1}{b_2} - {a_2}{b_1}}}  \\
+ \frac{{\sqrt {{{\left( {{a_1}{c_2}  -  {a_2}{c_1}} \right)}^2}   -  \left( {{a_1}{b_2} - {a_2}{b_1}} \right)\left( {{b_1} - {b_2}} \right){c_1}} }}{{{{a_1}{b_2} - {a_2}{b_1}}}},
 \end{multline}
 where
\begin{align}\label{eq:a1b1c1}
{a_1} = -{\gamma ^2}\left( {\left( {{N_t} - K} \right){\theta _m} - \left( {{N_t} - K} \right){\theta _{b,p}} - K{\theta _{b,q}}} \right){K{\theta _{e,q}}}
\end{align}
\begin{multline}
{b_1} = \gamma \left( {\left( {{N_t} - K} \right){\theta _m} + \left( {{N_t} - K} \right){\theta _{b,p}} - K{\theta _{b,q}}} \right) \\
  \times \left( {\gamma {\theta _{e,q}} + {N_t} - K} \right) \\ + \gamma \left( {\gamma {\theta _{b,q}} + {N_t} - K} \right) \left( {\left( {{N_t} - K} \right){\theta _{e,q}} - K{\theta _{e,q}}} \right)
 \end{multline}
\begin{align}
{c_1}    = \left( {\gamma {\theta _{b,q}} + {N_t} - K} \right)\left( {\gamma {\theta _{e,q}} + {N_t} - K} \right)
\end{align}
\begin{align}\label{eq:a2b2c2}
 {a_{2}}   = {\gamma ^2}\left( {\left( {{N_t} - K} \right){\theta _{b,p}} - K{\theta _{b,q}}} \right)  \left( {\left( {{N_t} - K} \right){\theta _{e,e}}  - K{\theta _{e,q}}} \right)
 \end{align}
 \begin{multline}
 {b_{2}}   = \gamma \left( {\gamma {\theta _{b,q}} + {N_t} - K} \right)  \left( {\left( {{N_t} - K} \right){\theta _{e,e}}  - K{\theta _{e,q}}} \right) \\
\\  + \gamma \left( \left( {{N_t} - K} \right) {\theta _{b,p}} - K{\theta _{b,q}} \right) \left( {\gamma {\theta _{e,q}} + {N_t} - K} \right)
  \end{multline}
\begin{align}
{c_{2}} = \left( {\gamma {\theta _{b,q}} + {N_t} - K} \right)\left( {\gamma {\theta _{e,q}} + {N_t} - K} \right),
\end{align}
 with
\begin{align}
{\theta _{e,e}}  = \frac{\Lambda _E^0}{\tr\left( {\widehat {\bf{R}}_{0m}^0} \right)}, \label{eq:theta_ee}
\end{align}
\begin{multline}
 {\theta _{e,q}}  =  \sum\limits_{l = 0}^L \tr\left( {{\bf{R}}_{E,T}^l} \right) - \sum\limits_{l = 0}^L {\sum\limits_{k = 1,k \ne m}^K {\tr{{\left( {\widehat {\bf{R}}_{lk}^l} \right)}^{ - 1}}\tr\left( {{\bf{R}}_{E,T}^l\widehat {\bf{R}}_{lk}^l} \right)} } \\
  - \sum\limits_{l = 0}^L {\tr{{\left( {\widehat {\bf{R}}_{lm}^l} \right)}^{ - 1}}\Lambda _E^l}
\end{multline}
and
  \begin{multline}
\Lambda _E^l =  {\tau ^2}{P_E}{\left| {\tr\left( {{\bf{C}}_{lm}^l{\bf{R}}_{E,T}^l} \right)} \right|^2}  + \tau {N_0}\tr\left( {{\bf{R}}_{E,T}^l{\bf{C}}_{lm}^l{{\left( {{\bf{C}}_{lm}^l} \right)}^H}} \right) \\
 + {\tau ^2} \sum\limits_{t = 0}^L {{P_{tm}}\tr\left( {{\bf{R}}_{E,T}^l{\bf{C}}_{lm}^l{\bf{R}}_{tm}^l{{\left( {{\bf{C}}_{lm}^l} \right)}^H}} \right)}.
\end{multline}
 Then, the optimal power allocation $p$ maximizing the asymptotic achievable secrecy rate in (\ref{eq:asy_sec}) for $N_e = 1$
 is given in (\ref{eq:p_opt}) at the top of the next page,  where ${R_{\sec ,\, {\rm asy}}}\left( p\right)$ is defined in (\ref{eq:asy_sec}).
  \begin{figure*}[!ht]
 \begin{align}\label{eq:p_opt}
 p^{*} = \left\{ \begin{array}{lll}
 1, &{ {\rm if}}&    {p_1} \notin \left[ {0,1} \right],{p_2} \notin \left[ {0,1} \right], \\
\arg \mathop { \max } \left\{ {R_{\sec ,\, {\rm asy}}}\left( 1\right),  {R_{\sec ,\, {\rm asy}}}\left( p_1\right) \right\},  &{ {\rm if}}&    {p_1} \in \left[ {0,1} \right],{p_2} \notin \left[ {0,1} \right] , \\
\arg \mathop {\max}  \left\{ {R_{\sec ,\, {\rm asy}}}\left( 1\right),  {R_{\sec ,\, {\rm asy}}}\left( p_1\right) \right\},  &{ {\rm if}}&   {p_1}
\notin \left[ {0,1} \right],{p_2} \in \left[ {0,1} \right], \\
\arg \mathop {\max}  \left\{ {R_{\sec ,\, {\rm asy}}}\left( 1\right),  {R_{\sec ,\, {\rm asy}}}\left( p_1\right), {R_{\sec ,\, {\rm asy}}}\left( p_2\right) \right\},  &{ {\rm if}}&   {p_1} \in \left[ {0,1} \right], {p_2} \in \left[ {0,1} \right],
\end{array} \right.
 \end{align}
 \hrulefill
\vspace*{4pt}
\end{figure*}
\begin{proof}
Please refer to Appendix \ref{proof:prop:opt_power}.
\end{proof}
\end{theorem}

{\bl A fundamental question in secure communication is  under which conditions a positive secrecy rate is achievable.
For massive MIMO systems under the pilot contamination attack, this question
is answered in the following theorem. }

\begin{theorem} \label{prop:sec_cond}
For $N_e$ = 1, to achieve ${R_{\sec ,\, {\rm asy}}} > 0$ in (\ref{eq:asy_sec}), the power allocated to the transmit signal must satisfy\footnote{When ${{a_1} - {a_2}} = 0$, if $b_1 - b_2 > 0$, then secure transmission can be achieved for
any $p$; otherwise, secure transmission cannot be achieved regardless of the value of $p$.}:
 \begin{align}
p & > - \frac{{ {{b_1} - {b_2}} }}{{ {{a_1} - {a_2}} }},  \, {\rm if} \,  {{a_1} - {a_2}}  > 0 \label{eq:sec_cond}  \\
 p & <  - \frac{{ {{b_1} - {b_2}} }}{{ {{a_1} - {a_2}} }},    \,  {\rm if} \  {{a_1} - {a_2}}  < 0.  \label{eq:sec_cond_1}
 \end{align}
  \begin{proof}
  Eqs. (\ref{eq:sec_cond}) and (\ref{eq:sec_cond_1}) can be  obtained by following a similar approach as in Appendix  \ref{proof:prop:opt_power} and
  finding the solution to $\frac{{1 + {\rm SINR}_{0m,\,{\rm asy}}}} {{ 1 + {\rm SINR}_{\rm eve,\,{\rm asy}} }} > 1$.
  \end{proof}
\end{theorem}

In the following, we derive a threshold which allows us to determine whether the MF-AN design
or the NS design achieves a higher secrecy rate.  Before proceeding, we introduce some definitions.  Define
${a_3} = q^*{\theta _{e,q}}\left( {p^*{\theta _{b,p}} + q^*{\theta _{b,q}} }  + p^*{\theta _m}\right)$,
${b_3} = q^*{\theta _{e,q}} + p^*{\theta _{b,p}} + q^*{\theta _{b,q}} + p^*{\theta _m}$,
${a_4} = \left( {p^*{\theta _{b,p}} + q^*{\theta _{b,q}}} \right)\left( {p^*{\theta _{e,e}} + q^*{\theta _{e,q}}} \right)$,
${b_4} = p^*{\theta _{b,p}} + q^*{\theta _{b,q}} + p^*{\theta _{e,e}} + q^*{\theta _{e,q}}$, where $p^*$ and $q^*$ are the optimal
power allocation values obtained from Theorem \ref{prop:opt_power}, ${a_5} =\frac{1}{K} \left({\theta _{b,p,{\rm{null}}}} + {\theta _{m,{\rm{null}}}}\right)$,
and ${a_6} = \frac{1}{K}{\theta _{b,p,{\rm{null}}}}$, where
\begin{multline}\label{eq:theta_m_null}
 {\theta _{m,\rm{null}}}  = \tr\left( {\widehat {\bf{R}}_{0m,\rm{null}}^0} \right) \\
 + \tr\left( {\widehat {\bf{R}}_{0m,\rm{null}}^0} \right)^{-1} {{\tr\left({\left({{\bf{R}}_{0m,\rm{null}}^0 - \widehat{\bf{R}}_{0m,\rm{null}}^0} \right)  \widehat{\bf{R}}_{0m,\rm{null}}^0} \right)}}
 \end{multline}
and
 \begin{multline}\label{eq:theta_bp_null}
{\theta _{b,p,{\rm{null}}}}  =  \sum\limits_{l = 0}^L {\sum\limits_{k = 1,k \ne m}^K {\tr{{\left( {\widehat {\bf{R}}_{lk}^l} \right)}^{ - 1}}
\tr\left( {{\bf{R}}^l_{0m}}\widehat {\bf{R}}_{lk}^l \right)}} \\
 +  \sum\limits_{l = 1}^L { \tr{{\left( {\widehat {\bf{R}}_{lm,{\rm{null}}}^l} \!\right)}^{ - 1}}  \Lambda _{0m}^l}.
 \end{multline}
Furthermore, define $a_7 = a_3  a_6 - a_4 a_5$, $b_7 = a_6  b_3 + a_3 - a_5 b_4 - a_4$, $c_7 = a_6 + b_3 - a_5 - b_4$, two fixed point equations for
$\gamma$
\begin{align}\label{eq:gamma_u_1}
\gamma  = \min \left[ {\frac{{ - {b_7} + \sqrt {b_7^2 - 4{a_7}{c_7}} }}{{2{a_7}}}, \frac{{ - {b_7} - \sqrt {b_7^2 - 4{a_7}{c_7}} }}{{2{a_7}}}} \right]
\end{align}
\begin{align}\label{eq:gamma_u_2}
\gamma = \max \left[ {\frac{{ - {b_7} + \sqrt {b_7^2 - 4{a_7}{c_7}} }}{{2{a_7}}}, \frac{{ - {b_7} - \sqrt {b_7^2 - 4{a_7}{c_7}} }}{{2{a_7}}}} \right],
\end{align}
and $\Delta  = b_7^2 - 4{a_7}{c_7}$.  We note that (\ref{eq:gamma_u_1}) and (\ref{eq:gamma_u_2})
are two fixed point equations since $a_3$, $b_3$, $a_4$, and $b_4$ are functions of $\gamma$.
Denote the solution of (\ref{eq:gamma_u_1}) and (\ref{eq:gamma_u_2}) by
${\gamma _{{\rm t}, \, 1}}$ and ${\gamma _{{\rm t}, \, 2}}$, respectively, and define
\begin{equation}\label{eq:beta_gamma}
\beta \left( \gamma  \right) = \left\{ \begin{array}{l}
1, \, {\rm if}   \  \left( \Delta  < 0{\rm{ } } \ {\rm and} \ {\rm{ }}{a_7} < 0 \right) \rm{or} \\
 \left(\Delta  > 0,{a_7} > 0  \ {\rm and} \   {{\gamma _{{\rm t}, \, 1}}} < \gamma  < {{\gamma _{{\rm t}, \, 2}}} \right)   {\rm or} \ \\
 \left( \Delta  > 0,{a_7} < 0{\rm{ }} \ {\rm and} \ \left( {\gamma  < {{\gamma _{{\rm t}, \, 1}}}{\rm{ }}
 \ {\rm{ }} \gamma  > \rm{or} \rm{or}  {{\gamma _{{\rm t}, \, 2}}}} \right)  \right)  \\
0,  \, {\rm otherwise}
\end{array} \right.
\end{equation}
Then, we have the following theorem.

 \begin{theorem} \label{prop:unified}
For $N_t \rightarrow \infty$ and $N_e = 1$, if $\beta \left( \gamma  \right) = 0$,  the asymptotic secrecy rate of the MF-AN design
is higher than that of  the NS design; if $\beta \left( \gamma  \right) = 1$,  the opposite is true.
\begin{proof}
Please refer to Appendix \ref{proof:prop:unified}.
\end{proof}
\end{theorem}

{\bl
\subsection{Single-Cell Single-User Case}
To provide more insights into  the impact of
the pilot contamination attack on secure communication
in massive MIMO systems, we simplify the system model to the single-cell
single-user case in this subsection.
The transmitter (Alice) sends the  desired signal to
the receiver (Bob), and an eavesdropper (Eve)
is present to overhear the signal, i.e., we have $L = 0$, $K = 1$, and $N_e = 1$ \cite{Khisti2010TIT,Khisti2010TIT_2,Oggier2011TIT}.
We investigate again under which condition  a positive secrecy rate is achievable.
With the simplified model, in the following theorem, we  obtain a more intuitive condition
than the one in Theorem \ref{prop:sec_cond}.

\begin{theorem} \label{prop:sec_cond_single_user}
For a single-cell single-user single-antenna eavesdropper system ($L = 0, K = 1, N_e = 1$),
to achieve ${R_{\sec ,\, {\rm asy}}} > 0$ in (\ref{eq:asy_sec}), the following condition must be
satisfied:
\begin{align} \label{eq:condition_single_theo}
 & (p - 1)\gamma \eta_1     < \eta_2,
\end{align}
where
\begin{multline} \label{eq:coff}
\eta_1  = \left[\left( {\tr^2\left( {\widehat {\bf{R}}_{01}^0} \right) + \tr\left( {\left( {{\bf{R}}_{01}^0 - \widehat {\bf{R}}_{01}^0} \right)\widehat {\bf{R}}_{01}^0} \right)} \right) \right. \\
\times \left( {\tr\left( {{\bf{R}}_{E,T}^0} \right) \tr\left( {\widehat {\bf{R}}_{01}^0} \right) - \Lambda } \right)   \\
\left. - \left( {\tr\left( {\left( {{\bf{R}}_{01}^0 - \widehat {\bf{R}}_{01}^0} \right)\widehat {\bf{R}}_{01}^0} \right) + \tr\left( {\widehat {\bf{R}}_{01}^0} \right)\tr\left( {{\bf{R}}_{01}^0 - \widehat {\bf{R}}_{01}^0} \right)} \right)\Lambda \right] \\
\end{multline}
\begin{multline} \label{eq:coff_2}
\eta_2    = \\
\left[\left( {\tr^2\left( {\widehat {\bf{R}}_{01}^0} \right) + \tr\left( {\left( {{\bf{R}}_{01}^0 - \widehat {\bf{R}}_{01}^0} \right)\widehat {\bf{R}}_{01}^0} \right)} \right) - \Lambda  \right] \tr\left( {\widehat {\bf{R}}_{01}^0} \right)
 \end{multline}
\begin{multline} \label{eq:lambda}
\Lambda  = {\tau ^2}P_{01}^2 \tr\left( {{{\bf{\Omega }}^H}{\bf{R}}_{E,T}^0{\bf{\Omega R}}_{01}^0} \right) + {\tau ^2}{P_{01}}{P_E}{\left| {\tr\left( {{\bf{\Omega R}}_{E,T}^0} \right)} \right|^2} \\
+ \tau {P_{01}}{N_0} \tr\left( {{{\bf{\Omega }}^H}{\bf{R}}_{E,T}^0{\bf{\Omega }}} \right)
\end{multline}
\begin{align}
{\bf{\Omega }} = {\bf{R}}_{01}^0{\left( {{N_0}{{\bf{I}}_{{N_t}}} + \tau \left( {{P_{01}}{\bf{R}}_{01}^0 + {P_E}{\bf{R}}_{E,T}^0} \right)} \right)^{ - 1}}.
\end{align}
\begin{proof}
Please refer to Appendix \ref{proof:prop:sec_cond_single_user}.
\end{proof}
\end{theorem}

{\emph{Remark 5:}} From (\ref{eq:condition_single_theo}), we obtain that coefficients $\eta_1$ and $\eta_2$
determine whether secrecy can be achieved. In the high SNR regime
when $\gamma \rightarrow \infty$, if $\eta_1 > 0$,  (\ref{eq:condition_single_theo}) holds for arbitrary $p < 1$
and secure communication can be achieved; if $\eta_1 < 0$,  (\ref{eq:condition_single_theo}) holds
only for $p > 1$.  However, $p$ denotes
the fraction of power allocated to the transmit signal and satisfies $0<p\leq 1$.
This means that if $\eta_1 < 0$,  (\ref{eq:condition_single_theo}) can not hold
and secure communication can not be achieved.

In the following, we provide a concrete example for when $\eta_1 < 0$ may occur.
\begin{theorem} \label{prop:sec_cond_single_user_iid}
 Consider i.i.d. fading, i.e., ${{\bf{R}}^{0}_{01}} = {\beta _{01}}{{\bf{I}}_{{N_t}}}$ and ${{\bf{R}}_{E,T}^{0}} = {\beta _E}{{\bf{I}}_{{N_t}}}$, where ${\beta _{01}}$
 and ${\beta _E}$ denote the path-losses for the desired user and the eavesdropper, respectively \cite{Zhou2012TWC,Kapetanovic2013}. Then,
 $\eta_1 \gtrless 0$ is equivalent to ${P_{01}}{\beta _{01}} \gtrless {P_E}{\beta _E}$.
\begin{proof}
Please refer to Appendix \ref{proof:prop:sec_cond_single_user_iid}.
\end{proof}
\end{theorem}

Theorems \ref{prop:sec_cond_single_user} and \ref{prop:sec_cond_single_user_iid} imply
that if the active eavesdropper increases the power of the pilot contamination attack,
the classical MF-AN design \cite{Zhu2014} may not be able  to achieve secure communication.
Hence, a new design is needed to facilitate secure communication in  massive MIMO systems.

\textit{Remark 6:} If $\eta_1 = 0$ and $\eta_2 > 0$, then (\ref{eq:condition_single_theo}) is always satisfied.
Hence, secure communication is achievable regardless of the pilot contamination power $P_E$ and SNR $\gamma$.
The question is how
$\eta_1 = 0$ and $\eta_2 > 0$ can be achieved.
Clearly, in a conventional massive MIMO setting with i.i.d. channels,
as in \cite{Zhou2010TVT}, this condition is satisfied with low probability,
 because for small-dimensional MIMO channels, both ${{\bf{R}}^{0}_{01}}$
and  ${{\bf{R}}_{E,T}^{0}}$ in (\ref{eq:coff}) are full rank positive definite matrices,
such that  $\eta_1 = 0$ in (\ref{eq:coff}) is difficulty
to achieve. However, a unique feature of
massive MIMO channels is that the transmit correlation matrices of the channels are
low rank \cite{Yin2013JSAC,Adhikary2013TIT,Yin2014JSTSP,Nam2014JSTSP,Adhikary2014JSAC,Sun2015TCOM,Shen2015CL}.
Therefore, for massive MIMO systems, it is possible to  perform
joint uplink and downlink processing as in  Section III-B to project
the channels of the users and the eavesdropper into the
null space of ${{\bf{R}}_{E,T}^{0}}$.
For this NS design,  the condition in (\ref{eq:condition_single_theo}) still
holds after replacing $\mathbf{R}_{01}^{0}$ and $\mathbf{R}_{E,T}^{0}$ in (\ref{eq:coff}) and
(\ref{eq:coff_2}) with $\mathbf{R}_{01,\rm{null}}^{0}$ and $\mathbf{0}$, respectively.
In this case, we obtain $\eta_1 = 0$ and $\eta_2 > 0$.
In fact, this is the intuition behind the NS design proposed in this paper.
Unlike the conventional
NS design for the perfect CSI case \cite{Khisti2010TIT}, the proposed NS design
transmits along the  statistical eigen-direction of the channel. Moreover, since the pilot contamination attack
affects both the uplink channel estimation and the downlink data transmission, the proposed NS design requires a joint
uplink and downlink processing to completely eliminate the impact of the pilot contamination attack,
see Section III-B for details.

}

\section{Numerical Results}
In this section, we provide numerical results to evaluate the secrecy performance of the considered massive MIMO system with
an active eavesdropper. We consider a system where a  uniform linear array with $N_t = 128$ is employed at the BS, the antenna
spacing is  half a wavelength, and the angle of arrival (AoA) interval is $\mathcal{A} = [-\pi,\pi]$. We use the truncated Laplacian distribution
to model the channel power angle spectrum as \cite{Cho2010}
\begin{align} \label{eq:p_theta}
p\left( \theta  \right) = \frac{1}{{\sqrt 2 \sigma \left( {1 - {e^{ - \sqrt 2 \pi /\sigma }}} \right)}}{e^{\frac{{ - \sqrt 2 \left\| {\theta  - \overline \theta  } \right\|}}{\sigma }}},
\end{align}
where $\sigma$ and $\overline\theta$ denote the angular spread (AS) and the mean  AoA of the channel, respectively.
We assume that the AS $\sigma$  in (\ref{eq:p_theta}) is identical for the channels of all  users and the eavesdropper and we
set $\sigma = \pi/2$.
The channel transmit correlation matrices of all users, $\mathbf{R}^l_{lk}$, and the eavesdropper, $\mathbf{R}^l_{E,T}$, are generated based on \cite[Eq. (3.14)]{Cho2010}.
For the channel between the user
and the BS in its own cell and the channel between the user
and the BSs in the other cells, we impose a channel power normalization
to make the trace of the channel transmit correlation matrices equal to $N_t$ and $\rho N_t$, respectively, and
set $\rho = 0.1$. The receive correlation matrices of the eavesdropper, $\mathbf{R}^l_{E,R}$, are generated using the exponential correlation model
 $ \left\{\mathbf{R}^l_{E,R} \right\}_{i,j} = \varphi^{| i - j |},  \varphi \in (0,1)$, where $\varphi$ is generated at random.
 {\bl We note that for any precoder ${\bf{P}}_e$ satisfying (\ref{eq:optimal_precoder}), the resulting
$r_{E,R}^0$ in (\ref{eq:mmse_h0m}) and (\ref{eq:R0m0_est}) is the same.
In the simulations, we set ${{{\bf{p}}_s}} = \sqrt{\frac{1}{N_e}} {{\bf{u}}_e}$, $s = 1,2,\cdots,N_e$.}
For the channel between the eavesdropper
and the $0$th BS  and the channel between the eavesdropper
and the other BSs, we impose a channel power normalization
to make the trace of the channel receive correlation matrices equal to $N_e$ and $\rho N_e$, respectively.
The asymptotic secrecy rate is computed based on
Theorem  \ref{prop:sec_rate_mul} and the exact secrecy rate is obtained by Monte Carlo
simulation.  We set $L = 3$, $K = 5$, $P_{lk} =1$, $k = 1,2,\cdots,K$, $l = 0,1,\cdots,L$, $\tau = 10$, and $N_0 = 1$.
The mean channel AoAs, $\overline \theta$, of all  users and the eavesdropper in (\ref{eq:p_theta}) are generated at random and
the channel AoAs, $ \theta$, of all users and the eavesdropper
are distributed within the angle interval $\left[-\pi,  \pi\right]$ according to (\ref{eq:p_theta}).

\begin{figure}[!t]
\centering
\includegraphics[width=0.5\textwidth]{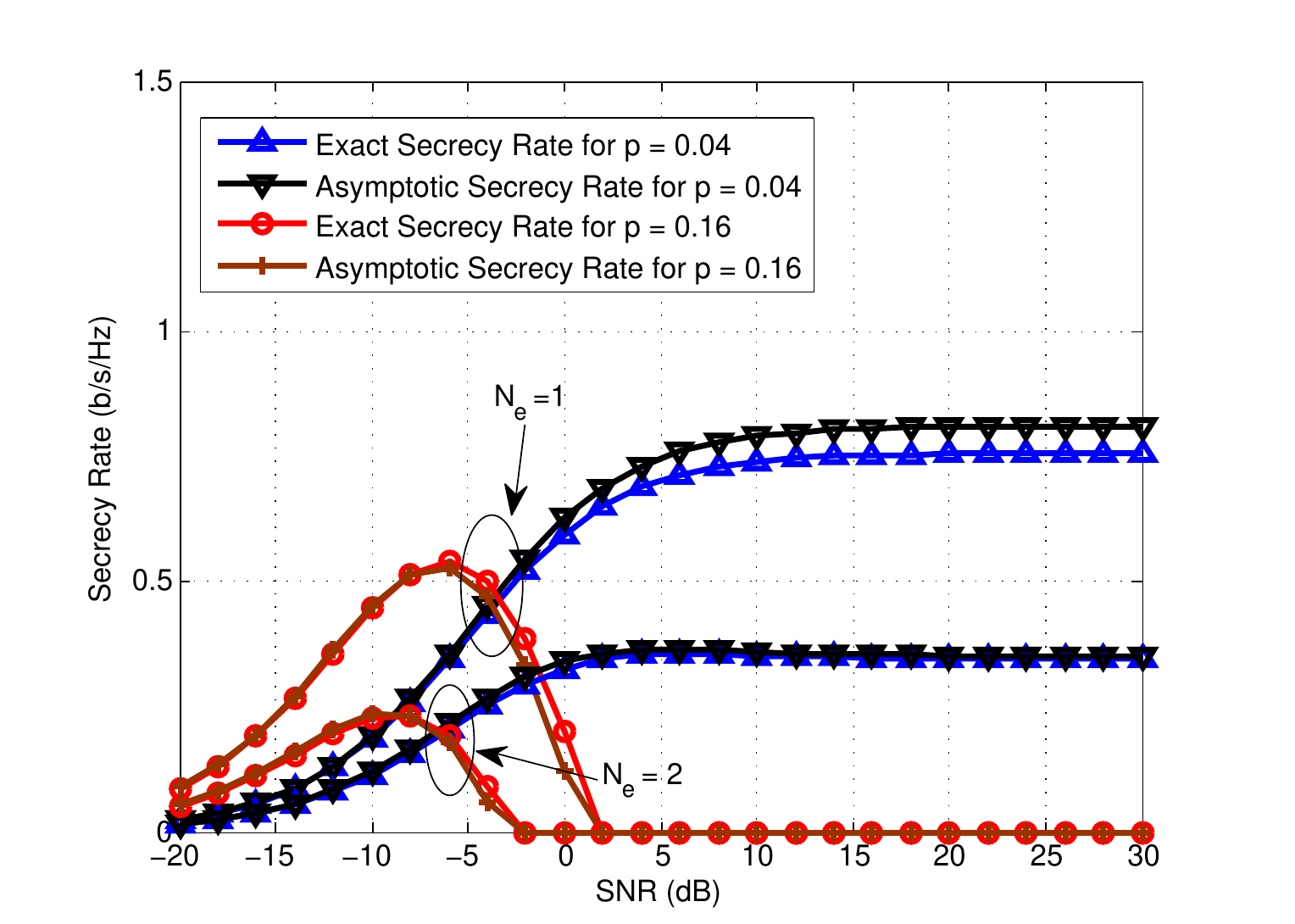}
\caption {\space\space  Secrecy rate vs. SNR for $P_E = 1$, different $p$, and different $N_e$.}
\label{Sec_Asy_SNR_Multi}
\end{figure}

Figure \ref{Sec_Asy_SNR_Multi} shows the secrecy rate performance versus (vs.)
SNR $\gamma$ for the MF-AN design, $P_E = 1$, different $p$, and different $N_e$.  We observe from Figure \ref{Sec_Asy_SNR_Multi} that the asymptotic secrecy rate
in Theorem \ref{prop:sec_rate_mul} provides a good estimate
for the exact secrecy rate. Also, we observe from Figure \ref{Sec_Asy_SNR_Multi} that in the low SNR regime,
allocating more power (larger $p$) to the information-carrying transmit signal leads to a higher secrecy rate. However, as the SNR increases,
the secrecy rate drops significantly if the transmit signal power is high and the AN power is small.
For example, for ${\rm SNR} = 2$ dB and $N_e = 1$,  we find from  Theorem \ref{prop:sec_cond}
that $p < 0.15$ is a necessary condition to guarantee reliable communication.
Thus, for $p = 0.16$, a positive secrecy rate cannot be achieved, as confirmed by Figure \ref{Sec_Asy_SNR_Multi}.
In addition, we observe from Figure \ref{Sec_Asy_SNR_Multi} that for the MF-AN design, increasing $N_e$ results in a secrecy rate degradation.

\begin{figure}[!t]
\centering
\includegraphics[width=0.5\textwidth]{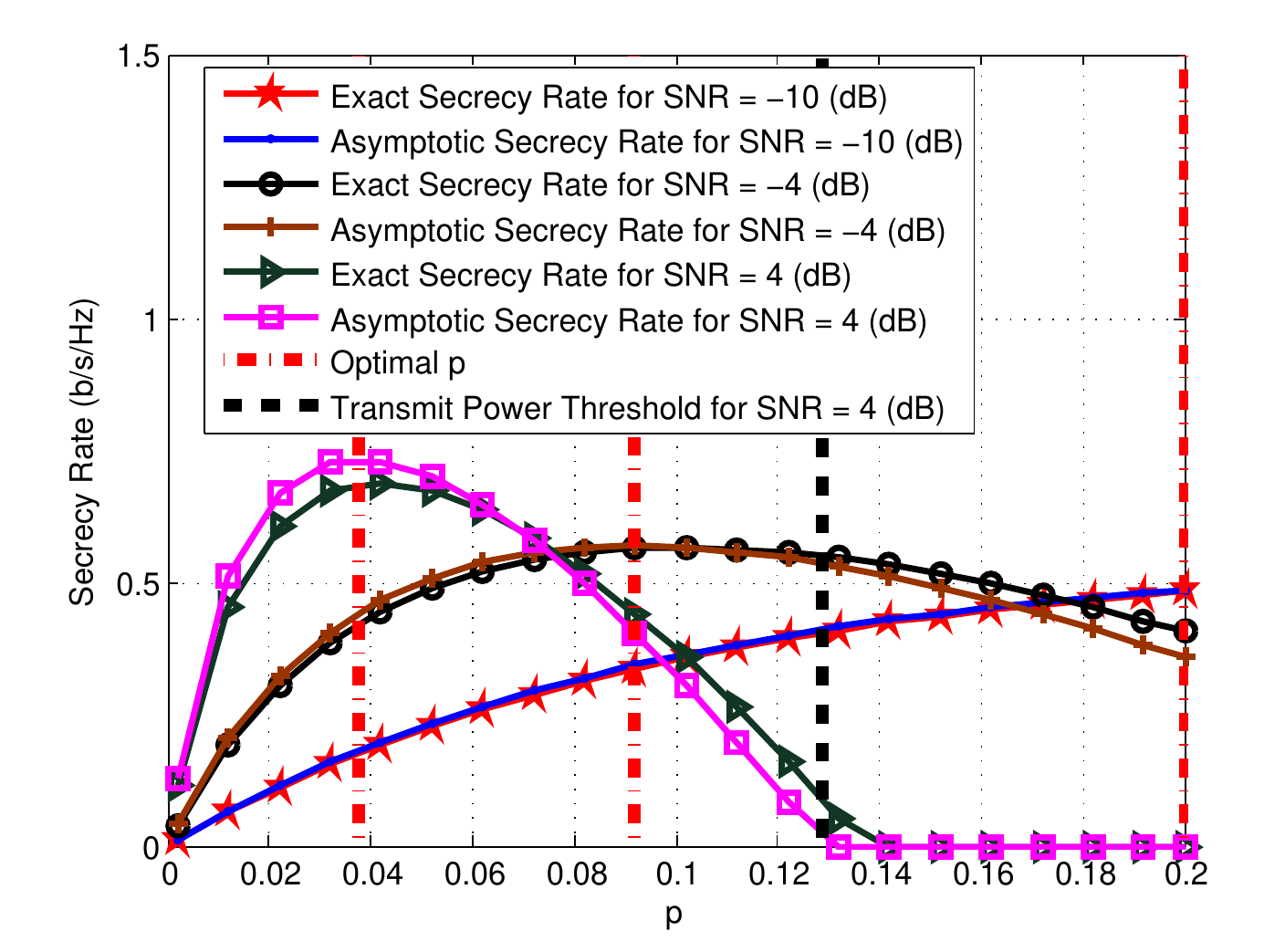}
\caption {\space\space  Secrecy rate vs. $p$ for $N_e = 1$, $P_E = 1$, and different SNRs.}
\label{Sec_Asy_P_Multi}
\end{figure}

Figure \ref{Sec_Asy_P_Multi} shows the secrecy rate performance vs. $p$
for the MF-AN design, $N_e = 1$, $P_E = 1$, and different SNRs.
For $K = 5$, we know from Section II-B that $0 \leq p \leq 0.2$.
Figure  \ref{Sec_Asy_P_Multi} confirms
that the asymptotic secrecy rate in Theorem \ref{prop:sec_rate_mul}
 provides a good estimate for the exact secrecy rate.  Also,
we observe from Figure  \ref{Sec_Asy_P_Multi} that the maximum of the exact
secrecy rate is achieved for the
optimal power allocation solution provided in Theorem \ref{prop:opt_power}.
Furthermore,  for ${\rm SNR} = 4$ dB, a positive secrecy rate cannot be achieved when $p$ is larger than the
analytical transmit power threshold provided in Theorem \ref{prop:sec_cond}.
As the SNR increases, more power has to be used for generating AN, in order
to achieve reliable secrecy transmission.

\begin{figure}[!t]
\centering
\includegraphics[width=0.5\textwidth]{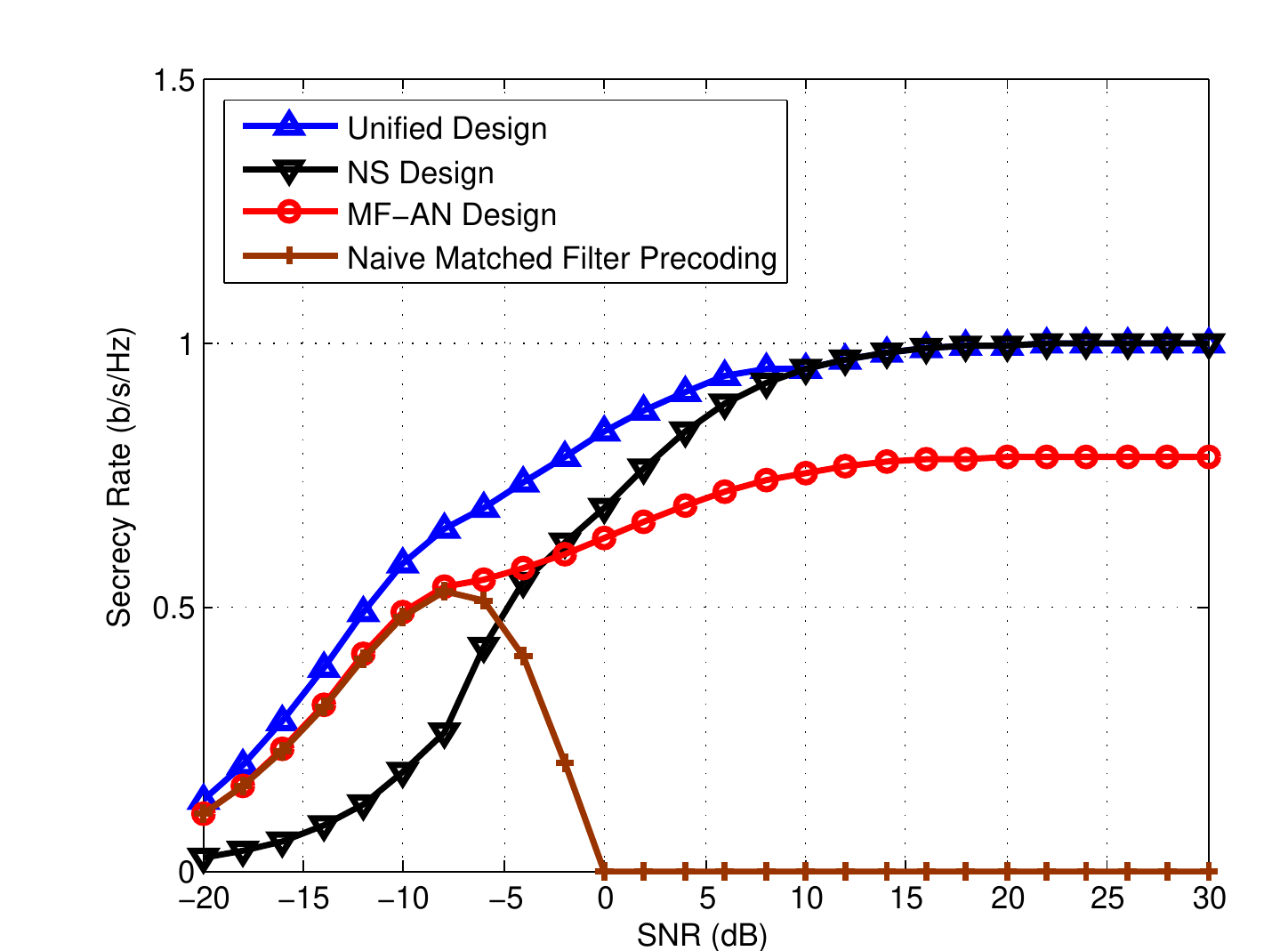}
\caption {\space\space  Exact secrecy rate vs. the SNR for $N_e = 1$, $P_E = 1$,  and different precoding designs.}
\label{Sec_Asy_SNR_Multi_Unified}
\end{figure}

Figure \ref{Sec_Asy_SNR_Multi_Unified} shows the exact secrecy rate performance vs. SNR $\gamma$ for $N_e = 1$, $P_E = 1$, and different system designs.
For the MF-AN design, we
adopted the optimal power allocation $p$ and $q$ based on Theorem \ref{prop:opt_power}.
For the NS design, when an eigenvalue of $\mathbf{R}_{E,T}^l$ is less than $10^{-3}$,
we treated it as being equal to zero. Based on this, we obtained  $T_l$ and $\mathbf{V}_E^l$ according to Section III-B.
For the unified design, we obtained the optimal $\alpha$ and $\beta$  by performing a
one-dimensional numerical search. For comparison, we also show results for a naive MF precoding scheme, where $p = 1/K$ and
AN is not generated.  From Figure \ref{Sec_Asy_SNR_Multi_Unified}, we make the following observations:
1) The unified design achieves the best performance for all considered SNR values.
2) Because of the pilot contamination attack, even though the transmitter is equipped with a large number of antennas,
naive MF precoding cannot achieve a positive secrecy rate for moderate-to-high SNRs.

\begin{figure}[!t]
\centering
\includegraphics[width=0.5\textwidth]{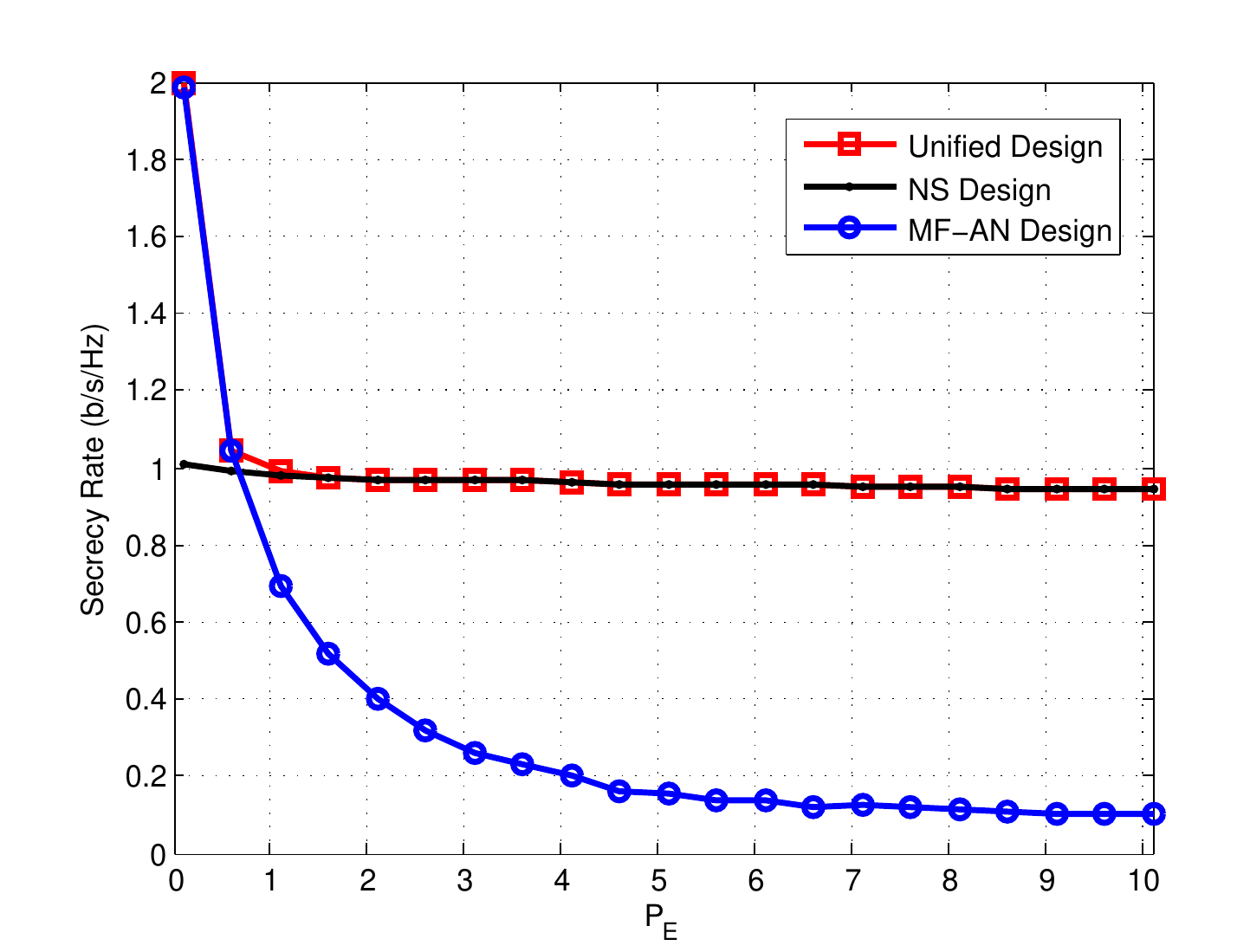}
\caption {\space\space  Exact secrecy rate vs. $P_E$ for $N_e = 1$, $\gamma = 10$ dB, and different precoding designs.}
\label{Sec_Asy_SNR_Multi_Unified_PE}
\end{figure}

Figure \ref{Sec_Asy_SNR_Multi_Unified_PE} shows
the exact secrecy rate performance vs. $P_E$  for $N_e = 1$, $\gamma = 10$ dB, and different system designs.
We observe from Figure \ref{Sec_Asy_SNR_Multi_Unified_PE}
that the unified design performs best for all considered $P_E$ values.
For a weak pilot contamination attack,
the MF-AN design performs better than the NS design.
However, when the eavesdropper increases its pilot power, this leads to a serious secrecy rate loss
for the MF-AN design, but has barely any impact on the NS design.
This is because the NS design can eliminate the
impact of the pilot contamination caused by the active eavesdropper as suggested by
(\ref{eq:R0m0_est_null}) where the term $P_E \mathbf{R}_{E,T}^{0}$ disappears in the covariance
matrix of the estimated channel. When $P_E = 0.1$ and $P_E = 1$,
Theorem \ref{prop:unified} indicates that, for $\gamma = 10$ dB,  the MF-AN design and the NS design
perform better, respectively, which is confirmed by Figure \ref{Sec_Asy_SNR_Multi_Unified_PE}.

\begin{figure}[!t]
\centering
\includegraphics[width=0.5\textwidth]{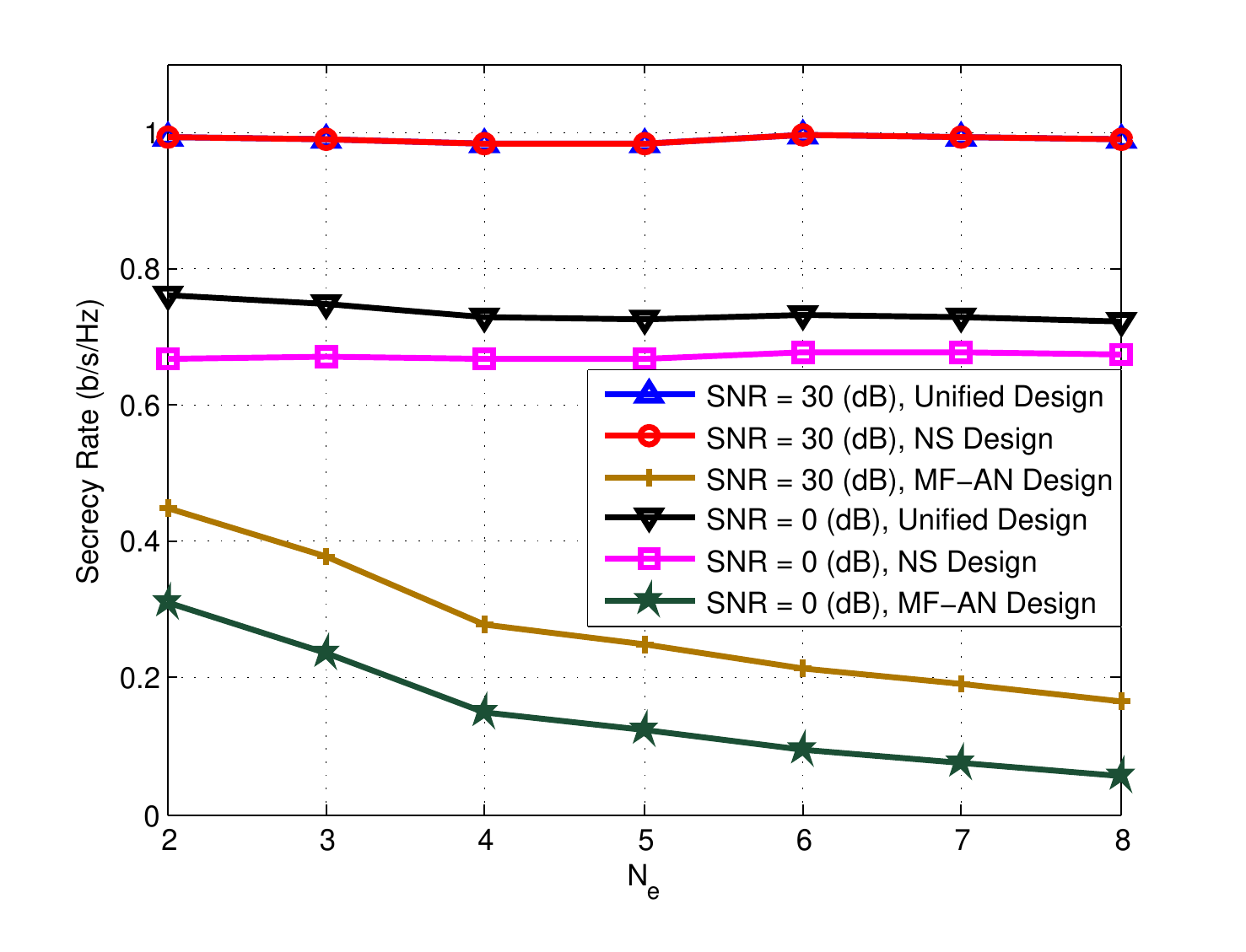}
\caption {\space\space  Exact secrecy rate vs. $N_e$ for $P_E = 1$,  different system designs, and different SNRs. }
\label{Sec_Asy_SNR_Multi_Unified_MulEve_MultiNr}
\end{figure}

Figure \ref{Sec_Asy_SNR_Multi_Unified_MulEve_MultiNr} shows the exact secrecy rate performance vs. $N_e$ for $P_E = 1$, different
system designs, and different SNRs.
We observe from Figure \ref{Sec_Asy_SNR_Multi_Unified_MulEve_MultiNr} that
the unified design performs best for the considered SNR values.
Also, the secrecy rate of the unified design barely
decreases with increasing $N_e$.
This confirms that the unified design is an effective approach for combating
the pilot contamination attack.  In contrast,
the secrecy rate performance of the MF-AN design degrades
significantly with increasing $N_e$.

\section*{Acknowledgment}
The authors would like to thank Prof. Xiqi Gao, Prof. C. K. Wen, and
Mr. Shahram Zarei for the helpful discussion throughout the paper,
Prof. Shi Jin for pointing out \cite{Kapetanovic2013}, and Prof. Jiaheng Wang
for the helpful discussion on Theorem \ref{theo:optimal_attack}.
The authors would also like to thank the editor, Prof. Yingbin Liang,
and anonymous reviewers for helpful comments and suggestions that greatly improve
the quality of the paper.

\section{Conclusions}
In this paper, we have studied the transmit signal design for multi-cell multi-user massive
MIMO systems in the presence of a multi-antenna active eavesdropper. For
the MF-AN design, we obtained an asymptotic
achievable secrecy rate expression for the pilot contamination attack
when the number of transmit antennas tends to infinity.
Moreover, we proved that the impact of the active eavesdropper can be completely eliminated when
the transmit correlation matrices of the users and the eavesdropper are orthogonal.
This analytical insight motivated the development of
transmit signal designs that are robust against the pilot contamination attack.
Also, for the MF-AN design and a single-antenna eavesdropper,
we derived
closed-form expressions for the optimal power allocation policy
for the information-carrying transmit signal and the AN as well as the minimum transmit signal power required
to ensure secure transmission.  In addition, for the single-antenna eavesdropper case,
a decision threshold for determining whether the MF-AN design or the NS design
is preferable was provided. Monte Carlo simulation results showed that the derived analytical results are accurate
and confirmed the effectiveness of the proposed transmission schemes for combating
the pilot contamination attack.

\appendices

{\bl

\section{Proof of Theorem \ref{theo:optimal_attack}}\label{proof:theo:optimal_attack}
First, recalling the definition of ${\bf W}_e$, we have
\begin{align}\label{eq:power_con_attack}
{{\bf{W}}_e}{\bf{W}}_e^H = \tau {{\bf{1}}_{{N_e}}}{\bf{1}}_{{N_e}}^T.
\end{align}
Then, the constraint in (\ref{eq:pilot_attack_contr}) reduces to
\begin{align}\label{eq:power_con_attack_2}
\tr\left( {{\bf{P}}_e {\bf{1}}_{N_e}{{\bf{1}}_{N_e}^T}{{\bf{P}}_e^H}} \right) = \tr\left({ { \left(\sum\limits_{r = 1}^{{N_r}} {{\bf{p}}_r} \right)}
\left({\sum\limits_{r = 1}^{{N_r}} {{\bf{p}}_r}}\right)^H  } \right)   \le {N_e}.
\end{align}

Also, (\ref{eq:pilot_attack}) can be re-written as
\begin{align}
 \sum\limits_{r = 1}^{{N_r}} {\sum\limits_{s = 1}^{{N_r}} {{{\left\{ {{{\bf{P}}_e^H}{\bf{R}}_{E,R}^0{\bf{P}}_e} \right\}}_{rs}}} } & = \sum\limits_{r = 1}^{{N_r}} {\sum\limits_{s = 1}^{{N_r}} {{\bf{e}}_r^H{{\bf{P}}_e^H}{\bf{R}}_{E,R}^0{\bf{P}}_e {{\bf{e}}_s}} } \nonumber  \\
  & = \sum\limits_{r = 1}^{{N_r}} {\sum\limits_{s = 1}^{{N_r}} {{\bf{p}}_r^H{\bf{R}}_{E,R}^0{{\bf{p}}_s}} } \nonumber \\
  & = \left( {\sum\limits_{r = 1}^{{N_r}} {{\bf{p}}_r^H} } \right){\bf{R}}_{E,R}^0\left( {\sum\limits_{s = 1}^{{N_r}} {{{\bf{p}}_s}} } \right).  \label{eq:R_ER_3}
\end{align}
 Define ${\bf{p}} = \sum\nolimits_{r = 1}^{{N_e}} {{{\bf{p}}_r}}$. Considering (\ref{eq:power_con_attack_2}), (\ref{eq:R_ER_3}), the optimization problem in (\ref{eq:pilot_attack}) is
equivalent to
\begin{equation} \label{eq:pilot_attack_equ}
\mathop {\max }\limits_{\bf{p}} {{\bf{p}}^H}{\bf{R}}_{E,R}^0{\bf{p}}
\end{equation}
\begin{equation} \label{eq:pilot_attack_contr_equ}
{\rm s.t.} \quad  \tr\left({\bf p} {\bf p}^H \right) \leq N_e \nonumber .
\end{equation}

It is easy to show that the optimal solution of (\ref{eq:pilot_attack_equ}) is
${\bf{p}} = \sqrt{N_e} {\bf u}_e$, where ${\bf u}_e$
is the eigenvector corresponding to the largest eigenvalue of ${\bf{R}}_{E,R}^0$.
This completes the proof.}

\section{Proof of Theorem \ref{prop:sec_rate_mul}}\label{proof:prop:sec_rate_mul}
First, we calculate ${\rm SINR}_{0m}$ in (\ref{eq:SINR_0m}) for $N_t \rightarrow \infty$.
For the numerator of (\ref{eq:SINR_0m}), based on \cite[Corollary 1]{Evans2000TIT},  we have
 \begin{align} \label{eq:h0mw0m_2_app}
& \frac{1}{{{N_t}}} \left[ {{{\left| {{{\left( {{\bf{h}}_{0m}^0} \right)}^H}{{\bf{w}}_{0m}}} \right|}^2}} \right] \nonumber \\
  & \mathop  \to \limits^{{N_t} \to \infty} \frac{1}{{{N_t}}}  \left| {\widehat {\bf{h}}_{0m}^0} \right|^2
+ \frac{1}{{{N_t}}}\left[ {\frac{{{{\left( {\widehat {\bf{h}}_{0m}^0} \right)}^H}\left( {{\bf{R}}_{0m}^0 - \widehat {\bf{R}}_{0m}^0} \right)\widehat {\bf{h}}_{0m}^0}}{{{{\left\| {\widehat {\bf{h}}_{0m}^0} \right\|}^2}}}} \right] \nonumber \\
& \mathop  \to \limits^{{N_t} \to \infty} \frac{1}{{{N_t}}}\tr\left( {\widehat {\bf{R}}_{0m}^0} \right)
+ \frac{1}{{{N_t}}} \frac{{\tr\left({\left({{\bf{R}}_{0m}^0 - \widehat{\bf{R}}_{0m}^0} \right) \widehat{\bf{R}}_{0m}^0} \right)}}{{\tr\left( {\widehat {\bf{R}}_{0m}^0} \right)}}.
\end{align}

For the denominator of (\ref{eq:SINR_0m}), we have
\begin{align} \label{eq:h0mw0k_app}
\frac{1}{{{N_t}}} {{{\left| {{{\left( {{\bf{h}}_{0m}^0} \right)}^H}{{\bf{w}}_{0k}}} \right|}^2}}& ={\frac{{\left( {{\bf{h}}_{0m}^0} \right)}^H \widehat {\bf{h}}_{0k}^0 {\left( {\widehat {\bf{h}}_{0k}^0} \right)}^H
{\left( {{\bf{h}}_{0m}^0} \right)}
}{{{{\left( {\widehat {\bf{h}}_{0k}^0} \right)}^H}\widehat {\bf{h}}_{0k}^0}}} \nonumber \\
&\mathop  \to \limits^{{N_t} \to \infty }   {\frac{{{{\left( {\widehat {\bf{h}}_{0k}^0} \right)}^H}{\bf{R}}_{0m}^0\widehat {\bf{h}}_{0k}^0}}{{{{\left( {\widehat {\bf{h}}_{0k}^0} \right)}^H}\widehat {\bf{h}}_{0k}^0}}} \nonumber \\
&\mathop  \to \limits^{{N_t} \to \infty } \frac{{\tr\left( {{\bf{R}}_{0m}^0\widehat {\bf{R}}_{0k}^0} \right)}}{{\tr\left( {\widehat {\bf{R}}_{0k}^0} \right)}}.
\end{align}
Also, performing some simplifications, we obtain
\begin{multline} \label{eq:h0mUull0_app}
\frac{1}{{{N_t}}}\left[ {{{\left\| {{{\left( {{\bf{h}}_{0m}^0} \right)}^H}{{\bf{U}}_{{\rm null},\,0}}} \right\|}^2}} \right] \mathop  \to \limits^{{N_t} \to \infty }
 \frac{1}{{{N_t}}} {{\bf{R}}_{0m}^0} \\
 - \frac{1}{{{N_t}}}\sum\limits_{k = 1,k \ne m}^K {\tr{{\left( {\widehat {\bf{R}}_{0k}^0} \right)}^{ - 1}} {{{\left( {{\bf{h}}_{0m}^0} \right)}^H}\left( {\widehat {\bf{h}}_{0k}^0} \right){{\left( {\widehat {\bf{h}}_{0k}^0} \right)}^H}{\bf{h}}_{0m}^0}}  \\
  - \frac{1}{{{N_t}}}\tr{\left( {\widehat {\bf{R}}_{0m}^0} \right)^{ - 1}} {{{\left( {{\bf{h}}_{0m}^0} \right)}^H}\left( {\widehat {\bf{h}}_{0m}^0} \right){{\left( {\widehat {\bf{h}}_{0m}^0} \right)}^H}{\bf{h}}_{0m}^0}.
\end{multline}
For $m \ne k$, ${\bf{h}}_{0m}^0$ is independent of $\widehat {\bf{h}}_{0k}^0$. Hence, the
asymptotic expression for $ \frac{1}{{{N_t}}} {{{\left( {{\bf{h}}_{0m}^0} \right)}^H}\widehat {\bf{h}}_{0k}^0{{\left( {\widehat {\bf{h}}_{0k}^0} \right)}^H}{\bf{h}}_{0m}^0}$
is given by
\begin{align} \label{eq:h0mh0k}
 \frac{1}{{{N_t}}} {{{\left( {{\bf{h}}_{0m}^0} \right)}^H}\widehat {\bf{h}}_{0k}^0{{\left( {\widehat {\bf{h}}_{0k}^0} \right)}^H}{\bf{h}}_{0m}^0} \mathop  \to \limits^{{N_t} \to \infty} \frac{1}{{{N_t}}} \tr\left( {{\bf{R}}_{0m}^0\widehat {\bf{R}}_{0k}^0} \right).
\end{align}
For $m = k$, based on (\ref{eq:mmse_h0m}), we have
\begin{multline} \label{eq:h0m_app}
\widehat {\bf{h}}_{0m}^0  = \sqrt {{P_{0m}}} \tau {\bf{C}}_{0m}^0{\bf{h}}_{0m}^0 + {\bf{C}}_{0m}^0\sum\limits_{l = 1}^L {\sqrt {{P_{lm}}} \tau {\bf{h}}_{lm}^0} \\ + {\bf{C}}_{0m}^0  \sqrt {\frac{{{P_E}}}{{{N_e}}}} \tau \sum\limits_{r = 1}^{{N_e}} {{\bf{h}}_{E,r}^0}  + {\bf{C}}_{0m}^0  {{\bf{\Omega }}_m} {\bf{n}},
\end{multline}
where ${\bf{C}}_{0m}^0$ is defined in (\ref{eq:Clm}) and ${{\bf{\Omega }}_m} =  {\left( {{{\boldsymbol{\omega }}_m} \otimes {{\bf{I}}_{{N_t}}}} \right)^H}$.

Based on (\ref{eq:h0m_app}), we obtain
\begin{multline}\label{eq:h0mhath0m_app}
 {\left( {{\bf{h}}_{0m}^0} \right)^H}\widehat {\bf{h}}_{0m}^0{\left( {\widehat {\bf{h}}_{0m}^0} \right)^H}{\bf{h}}_{0m}^0 \\
= {\left( {{\bf{h}}_{0m}^0} \right)^H}{\bf{C}}_{0m}^0{{\bf{\Omega }}_m}{\bf{{n}}}{\left( {\widehat {\bf{h}}_{0m}^0} \right)^H}{\bf{h}}_{0m}^0  + {\tau ^2}{\left( {{\bf{h}}_{0m}^0} \right)^H}{\bf{C}}_{0m}^0 \\
 \times \sum\limits_{t = 0}^L {\sum\limits_{s = 0}^L {\sqrt {{P_{tm}}} \sqrt {{P_{sm}}} {\bf{h}}_{tm}^0} } {\left( {{\bf{h}}_{sm}^0} \right)^H}{\left( {{\bf{C}}_{0m}^0} \right)^H}{\bf{h}}_{0m}^0 +   {\tau ^2}\sqrt {\frac{{{P_E}}}{{{N_e}}}} \\
 \times {\left( {{\bf{h}}_{0m}^0} \right)^H}{\bf{C}}_{0m}^0  \sum\limits_{t = 0}^L {\sqrt {{P_{tm}}} {\bf{h}}_{tm}^0  \sum\limits_{r = 1}^{{N_e}} \left({{\bf{h}}_{E,r}^0}\right)^H
} {\left( {{\bf{C}}_{0m}^0} \right)^H}{\bf{h}}_{0m}^0 \\
 + \tau {\left( {{\bf{h}}_{0m}^0} \right)^H}{\bf{C}}_{0m}^0\sum\limits_{t = 0}^L {\sqrt {{P_{tm}}} {\bf{h}}_{tm}^0{{ {\bf{n}}}^H  {{\bf{\Omega }}_m^H}    }} {\left( {{\bf{C}}_{0m}^0} \right)^H}{\bf{h}}_{0m}^0  \\
+ \tau \sqrt {{P_E}} {\left( {{\bf{h}}_{0m}^0} \right)^H}{\bf{C}}_{0m}^0\sum\limits_{r = 1}^{{N_e}}{{\bf{h}}_{E,r}^0}{\left( {\widehat {\bf{h}}_{0m}^0} \right)^H}{\bf{h}}_{0m}^0.
\end{multline}

When $N_t \rightarrow \infty$, based on (\ref{eq:h0mhath0m_app}) and \cite[Corollary 1]{Evans2000TIT},  we have
\begin{align}
& \frac{1}{{{N_t}}}  {\left( {{\bf{h}}_{0m}^0} \right)^H}\widehat {\bf{h}}_{0m}^0{\left( {\widehat {\bf{h}}_{0m}^0} \right)^H}{\bf{h}}_{0m}^0 \nonumber \\
 & \mathop  \to \limits^{{N_t} \to \infty } \frac{1}{{{N_t}}}  {\tau ^2}{\left( {{\bf{h}}_{0m}^0} \right)^H}{\bf{C}}_{0m}^0\sum\limits_{t = 0}^L {{P_{tm}}{\bf{h}}_{tm}^0{{\left( {{\bf{h}}_{tm}^0} \right)}^H}{{\left( {{\bf{C}}_{0m}^0} \right)}^H}{\bf{h}}_{0m}^0} \nonumber  \\
& + \frac{1}{{{N_t}}} {\tau ^2}  {\frac{{{P_E}}}{{{N_e}}}}  {\left( {{\bf{h}}_{0m}^0} \right)^H}{\bf{C}}_{0m}^0 \sum\limits_{r = 1}^{{N_e}} {{\bf{h}}_{E,r}^0}
\sum\limits_{r = 1}^{{N_e}} \left({{\bf{h}}_{E,r}^0}\right)^H {\left( {{\bf{C}}_{0m}^0} \right)^H}{\bf{h}}_{0m}^0 \nonumber \\
&  + \frac{1}{{{N_t}}} {\left( {{\bf{h}}_{0m}^0} \right)^H}{\bf{C}}_{0m}^0{{\bf{\Omega }}_m}{\bf{n}}{\left( {{\bf{C}}_{0m}^0{{\bf{\Omega }}_m}{\bf{n}}} \right)^H}{\bf{h}}_{0m}^0  \nonumber  \\
&   \mathop  \to \limits^{{N_t} \to \infty }  \frac{1}{{{N_t}}}  {\tau ^2}{P_{0m}}{\left( {{\bf{h}}_{0m}^0} \right)^H}{\bf{C}}_{0m}^0{\bf{h}}_{0m}^0{\left( {{\bf{h}}_{0m}^0} \right)^H}{\left( {{\bf{C}}_{0m}^0} \right)^H}{\bf{h}}_{0m}^0 \nonumber \\
& + \frac{1}{{{N_t}}} {\tau ^2}{\left( {{\bf{h}}_{0m}^0} \right)^H}{\bf{C}}_{0m}^0\sum\limits_{t = 1}^L {{P_{tm}}{\bf{h}}_{tm}^0{{\left( {{\bf{h}}_{tm}^0} \right)}^H}{{\left( {{\bf{C}}_{0m}^0} \right)}^H}{\bf{h}}_{0m}^0} \nonumber \\
& + \frac{1}{{{N_t}}} {\tau ^2}  {\frac{{{P_E}}}{{{N_e}}}} {\left( {{\bf{h}}_{0m}^0} \right)^H}{\bf{C}}_{0m}^0\sum\limits_{r = 1}^{{N_e}} {{\bf{h}}_{E,r}^0}
\sum\limits_{r = 1}^{{N_e}} \left({{\bf{h}}_{E,r}^0}\right)^H{\left( {{\bf{C}}_{0m}^0} \right)^H}{\bf{h}}_{0m}^0 \nonumber \\
 & + \frac{1}{{{N_t}}} {\left( {{\bf{h}}_{0m}^0} \right)^H}{\bf{C}}_{0m}^0{{\bf{\Omega }}_m}{\bf{n}}{\left( {{\bf{C}}_{0m}^0{{\bf{\Omega }}_m}{\bf{n}}} \right)^H}{\bf{h}}_{0m}^0 \mathop  \to \limits^{{N_t} \to \infty }  \frac{1}{{{N_t}}}  \Lambda _{0m}^0. \label{eq:h0mhath0m_2_2_app}
\end{align}

Also, for $m \ne k$, ${\bf{h}}_{0m}^l$ is independent of $\widehat {\bf{h}}_{lk}^l$,
and we obtain in (\ref{eq:SINR_0m_b})
 \begin{align} \label{eq:h0mlwlk_app}
& \frac{1}{{{N_t}}} {{{\left| {{{\left( {{\bf{h}}_{0m}^l} \right)}^H}{{\bf{w}}_{lk}}} \right|}^2}} \nonumber \\
&  = \frac{1}{{{N_t}}} {{{\left( {{\bf{h}}_{0m}^l} \right)}^H}\frac{{\widehat {\bf{h}}_{lk}^l}}{{\left| {\widehat {\bf{h}}_{lk}^l} \right|}}\frac{{{{\left( {\widehat {\bf{h}}_{lk}^l} \right)}^H}}}{{\left| {\widehat {\bf{h}}_{lk}^l} \right|}}{\bf{h}}_{0m}^l} \nonumber \\
&\mathop  \to \limits^{{N_t} \to \infty } \frac{{{{\left( {\widehat {\bf{h}}_{lk}^l} \right)}^H}{\bf{R}}_{0m}^l\widehat {\bf{h}}_{lk}^l}}{{{{\left| {\widehat {\bf{h}}_{lk}^l} \right|}^2}}} \mathop  \to \limits^{{N_t} \to \infty } \frac{{\tr\left( {{\bf{R}}_{0m}^l\widehat {\bf{R}}_{lk}^l} \right)}}{{\tr\left( {\widehat {\bf{R}}_{lk}^l} \right)}}.
\end{align}

For $m = k$, similar to (\ref{eq:h0mhath0m_2_2_app}), we have
\begin{align} \label{eq:h0mhath0m_3_app}
\frac{1}{{{N_t}}}  {{{\left| {{{\left( {{\bf{h}}_{0m}^l} \right)}^H}{{\bf{w}}_{lm}}} \right|}^2}} & = \frac{1}{N_t}\frac{\left( {{\bf{h}}_{0m}^l} \right)^H  \widehat {\bf{h}}_{lm}^l  {\left( {\widehat {\bf{h}}_{lm}^l} \right)^H}\! {\bf{h}}_{0m}^l}{\left| {\widehat {\bf{h}}_{lm}^l} \right|^2} \nonumber \\
& \mathop  \to \limits^{{N_t}\to \infty  }  \frac{1}{{{N_t}}}  \frac{\Lambda _{0m}^l}{\tr\left( {\widehat {\bf{R}}_{lm}^l} \right)}.
\end{align}

Next, we simplify
\begin{align} \label{eq:h0m_null_app}
 & {{{\left( {{\bf{h}}_{0m}^l} \right)}^H}{{\bf{U}}_{{\rm null},\,l}}{{ {{{\bf{U}}_{{\rm null},\,l}^H}} }}{\bf{h}}_{0m}^l}= \nonumber  \\
  & \tr\left( {{\bf{R}}_{0m}^l} \right) - \sum\limits_{k = 1,k \ne m}^K {\tr{{\left( {\widehat {\bf{R}}_{lk}^l} \right)}^{ - 1}} {{{\left( {{\bf{h}}_{0m}^l} \right)}^H}\left( {\widehat {\bf{h}}_{lk}^l} \right){{\left( {\widehat {\bf{h}}_{lk}^l} \right)}^H}{\bf{h}}_{0m}^l} } \nonumber \\
& - \tr{\left( {\widehat {\bf{R}}_{lm}^l} \right)^{ - 1}} {{{\left( {{\bf{h}}_{0m}^l} \right)}^H}\left( {\widehat {\bf{h}}_{lm}^l} \right){{\left( {\widehat {\bf{h}}_{lm}^l} \right)}^H}{\bf{h}}_{0m}^l}.
\end{align}

Following a similar approach as was used to obtain (\ref{eq:h0mh0k}) and (\ref{eq:h0mhath0m_2_2_app}), we obtain
 \begin{align} \label{eq:h0mhlk}
& \frac{1}{{{N_t}}} {\left( {{\bf{h}}_{0m}^l} \right)^H}\left( {\widehat {\bf{h}}_{lk}^l} \right){\left( {\widehat {\bf{h}}_{lk}^l} \right)^H}{\bf{h}}_{0m}^l  \mathop  \to \limits^{{N_t}\to \infty  } \tr\left( {{\bf{R}}_{0m}^l\widehat {\bf{R}}_{lk}^l} \right)  \\
& \frac{1}{{{N_t}}} {{{\left( {{\bf{h}}_{0m}^l} \right)}^H}\left( {\widehat {\bf{h}}_{lm}^l} \right){{\left( {\widehat {\bf{h}}_{lm}^l} \right)}^H}{\bf{h}}_{0m}^l}   \mathop  \to \limits^{{N_t}\to \infty}
\Lambda _{0m}^l. \label{eq:h0mhlk_2}
\end{align}

By substituting (\ref{eq:h0mw0m_2_app})--(\ref{eq:h0mhlk_2}) into (\ref{eq:SINR_0m}), we obtain the expression for ${\rm SINR}_{0m,\,\rm{asy}}$
in (\ref{eq:b}).

Next, we simplify (\ref{eq:C_eve_1}). First, we have
\begin{multline} \label{eq:Q_l}
 {\left( {{\bf{H}}_E^l} \right)^H}{{\bf{U}}_{{\rm null},\,l}}{{{{\bf{U}}_{{\rm null},\,l}^H}}}{\bf{H}}_E^l   \\
 = {\left( {{\bf{H}}_E^l} \right)^H}  \left( {{{\bf{I}}_{{N_t}}} - \widehat {\bf{H}}_l^l
{\rm diag} \! {\left[{\tr\left( {\widehat {\bf{R}}_{l1}^l} \right)^{-1}},  \! {\tr\left( {\widehat {\bf{R}}_{l2}^l} \right)^{-1}},\!\cdots,\! \right.} } \right. \\
\left.{{\left.{\tr\left( {\widehat {\bf{R}}_{lK}^l} \right)^{-1} } \!\right]}
{{\left( {\widehat {\bf{H}}_l^l} \right)}^H}} \right){\bf{H}}_E^l   = {\left( {{\bf{H}}_E^l} \right)^H}{\bf{H}}_E^l - \mathbf{Q}_{H}
\end{multline}
where
\begin{multline} \label{eq:Q_H}
\mathbf{Q}_{H} = {\left( {{\bf{H}}_E^l} \right)^H} \widehat {\bf{H}}_l^l  \\
 \times {\rm diag} {\left[{\tr\left( {\widehat {\bf{R}}_{l1}^l} \right)^{-1}},  \! {\tr\left( {\widehat {\bf{R}}_{l2}^l} \right)^{-1}},\!\cdots,\! {\tr\left( {\widehat {\bf{R}}_{lK}^l} \right)^{-1} } \!\right]} {{\left( {\widehat {\bf{H}}_l^l} \right)}^H}{\bf{H}}_E^l.
\end{multline}

For ${\left( {{\bf{H}}_E^l} \right)^H}{\bf{H}}_E^l$, we have
\begin{align}
& \left\{{\left( {{\bf{H}}_E^l} \right)^H}{\bf{H}}_E^l\right\}_{ij} = {\left( {{\bf{h}}_{E,i}^l} \right)^H}{\bf{h}}_{E,j}^l \nonumber \\
& = {\bf{e}}_i^H{\left( {{\bf{R}}_{E,R}^l} \right)^{1/2}}{\left( {{\bf{G}}_E^l} \right)^H}{\left( {{\bf{R}}_{E,T}^l} \right)^{1/2}} \nonumber\\
& \hspace{3cm} \times {\left( {{\bf{R}}_{E,T}^l} \right)^{1/2}}{\bf{G}}_E^l{\left( {{\bf{R}}_{E,R}^l} \right)^{1/2}}{{\bf{e}}_j} \nonumber \\
& =  \tr\left( {{\bf{e}}_i^H{{\left( {{\bf{R}}_{E,R}^l} \right)}^{1/2}}{{\left( {{\bf{G}}_E^l} \right)}^H}{{\left( {{\bf{R}}_{E,T}^l} \right)}^{1/2}} } \right. \nonumber \\
& \hspace{3cm} \times \left.{{{\left( {{\bf{R}}_{E,T}^l} \right)}^{1/2}}{\bf{G}}_E^l{{\left( {{\bf{R}}_{E,R}^l} \right)}^{1/2}}{{\bf{e}}_j}} \right) \nonumber \\
& \mathop  \to \limits^{{N_t} \to \infty } \tr\left( {{{\left( {{\bf{R}}_{E,R}^l} \right)}^{1/2}}{{\bf{e}}_j}{\bf{e}}_i^H{{\left( {{\bf{R}}_{E,R}^l} \right)}^{1/2}}} \right)\tr\left( {{\bf{R}}_{E,T}^l} \right) \nonumber  \\
& = {\left\{ {{\bf{R}}_{E,R}^l} \right\}_{ij}}\tr\left( {{\bf{R}}_{E,T}^l} \right), \label{eq:h_E_ij_asy}
\end{align}
where we used \cite[Eq. (102)]{Wen2013TIT}.
Then, performing some simplifications, we get
\begin{align} \label{eq:Q_H}
{\left\{ {{{\bf{Q}}_H}} \right\}_{ij}} = \sum\limits_{k = 1}^K {\tr{{\left( {\widehat {\bf{R}}_{lk}^l} \right)}^{ - 1}}
{{\left( {{\bf{h}}_{E,i}^l} \right)}^H}} \left( {\widehat {\bf{h}}_{lk}^l} \right){\left( {\widehat {\bf{h}}_{lk}^l} \right)^H}{\bf{h}}_{E,j}^l.
\end{align}
For $m \ne k$, ${{\bf{h}}_{E,i}^l}$ and ${\bf{h}}_{E,j}^l$ are independent of $\widehat {\bf{h}}_{lk}^l$. Thus, we have
\begin{align}
& {\left( {{\bf{h}}_{E,i}^l} \right)^H}\widehat {\bf{h}}_{lk}^l{\left( {\widehat {\bf{h}}_{lk}^l} \right)^H}{\bf{h}}_{E,j}^l  ={\left( {\widehat {\bf{h}}_{lk}^l} \right)^H}{\bf{h}}_{E,j}^l{\left( {{\bf{h}}_{E,i}^l} \right)^H}\widehat {\bf{h}}_{lk}^l   \nonumber \\
& \mathop \to \limits^{{N_t} \to \infty }  \tr\left( {\widehat {\bf{R}}_{lk}^l{\bf{h}}_{E,j}^l{{\left( {{\bf{h}}_{E,i}^l} \right)}^H}} \right)  = {\left( {{\bf{h}}_{E,i}^l} \right)^H}\widehat {\bf{R}}_{lk}^l{\bf{h}}_{E,j}^l  \nonumber  \\
 & ={\bf{e}}_i^H{\left( {{\bf{R}}_{E,R}^l} \right)^{1/2}}{\left( {{\bf{G}}_E^l} \right)^H}{\left( {{\bf{R}}_{E,T}^l} \right)^{1/2}}\widehat {\bf{R}}_{lk}^l \nonumber \\
& \hspace{3cm} \times  {\left( {{\bf{R}}_{E,T}^l} \right)^{1/2}}{\bf{G}}_E^l{\left( {{\bf{R}}_{E,R}^l} \right)^{1/2}}{{\bf{e}}_j} \nonumber \\
&  \mathop  \to \limits^{{N_t} \to \infty } {\left\{ {{\bf{R}}_{E,R}^l} \right\}_{ij}}\tr\left( {{\bf{R}}_{E,T}^l\widehat {\bf{R}}_{lk}^l} \right), \label{eq:Q_H_3}
\end{align}
where (\ref{eq:Q_H_3}) is obtained based on \cite[Corollary 1]{Evans2000TIT} and \cite[Eq. (102)]{Wen2013TIT}.

For $m = k$,  similar to (\ref{eq:h0mhath0m_2_2_app}), we have
\begin{multline}
 {\left( {{\bf{h}}_{E,i}^l} \right)^H}\widehat {\bf{h}}_{lm}^l{\left( {\widehat {\bf{h}}_{lm}^l} \right)^H}{\bf{h}}_{E,j}^l \\
  \mathop \to \limits^{{N_t} \to \infty } {\tau ^2}\sum\limits_{t = 0}^L {{P_{tm}}} {\left( {{\bf{h}}_{E,i}^l} \right)^H}{\bf{C}}_{lm}^l{\bf{h}}_{tm}^l{\left( {{\bf{h}}_{tm}^l} \right)^H}{\left( {{\bf{C}}_{lm}^l} \right)^H}{\bf{h}}_{E,j}^l  \\
  + {\tau ^2}\frac{{{P_E}}}{{{N_e}}}\sum\limits_{r = 1}^{{N_e}} {\sum\limits_{t = 1}^{{N_e}} {{{\left( {{\bf{h}}_{E,i}^l} \right)}^H}{\bf{C}}_{lm}^l{\bf{h}}_{E,r}^l} } {\left( {{\bf{h}}_{E,t}^l} \right)^H}{\left( {{\bf{C}}_{lm}^l} \right)^H}{\bf{h}}_{E,j}^l \\
   + {\left( {{\bf{h}}_{E,i}^l} \right)^H}{\bf{C}}_{lm}^l{\bf{\Omega }}_m{ \bf{n}}{ {{\bf{ n}}}^H {\bf{\Omega }}_m^H}{\left( {{\bf{C}}_{lm}^l} \right)^H}{\bf{h}}_{E,j}^l.
\end{multline}
Then, based on \cite[Corollary 1]{Evans2000TIT} and \cite[Eq. (102)]{Wen2013TIT}, we obtain
\begin{align}
& {\left( {{\bf{h}}_{E,i}^l} \right)^H}{\bf{C}}_{lm}^l{\bf{h}}_{E,r}^l{\left( {{\bf{h}}_{E,t}^l} \right)^H}{\left( {{\bf{C}}_{lm}^l} \right)^H}{\bf{h}}_{E,j}^l \nonumber \\
& = {\bf{e}}_i^H{\left( {{\bf{R}}_{E,R}^l} \right)^{1/2}}{\left( {{\bf{G}}_E^l} \right)^H}{\left( {{\bf{R}}_{E,T}^l} \right)^{1/2}} \nonumber \\
 & \hspace{1cm} \times {\bf{C}}_{lm}^l{\left( {{\bf{R}}_{E,T}^l} \right)^{1/2}}{{{\bf{G}}_E^l}}{\left( {{\bf{R}}_{E,R}^l} \right)^{1/2}}{{\bf{e}}_r} \nonumber  \\
& \hspace{1.5cm} \ \times {\bf{e}}_t^H{\left( {{\bf{R}}_{E,R}^l} \right)^{1/2}}{\left( {{\bf{G}}_E^l} \right)^H}{\left( {{\bf{R}}_{E,T}^l} \right)^{1/2}}{\left( {{\bf{C}}_{lm}^l} \right)^H} \nonumber \\
& \hspace{2cm} \times {\left( {{\bf{R}}_{E,R}^l} \right)^{1/2}}{\bf{G}}_E^l{\left( {{\bf{R}}_{E,T}^l} \right)^{1/2}}{{\bf{e}}_j} \nonumber \\
 & \mathop  \to \limits^{{N_t} \to \infty } {\left\{ {{\bf{R}}_{E,R}^l} \right\}_{ir}}{\left\{ {{\bf{R}}_{E,R}^l} \right\}_{tj}}\tr\left( {{\bf{C}}_{lm}^l{\bf{R}}_{E,T}^l} \right)\tr\left( {{{\left( {{\bf{C}}_{lm}^l} \right)}^H}{{ {{\bf{R}}_{E,T}^l}}}} \right) \nonumber \\
& = {\left\{ {{\bf{R}}_{E,R}^l} \right\}_{ir}}{\left\{ {{\bf{R}}_{E,R}^l} \right\}_{tj}}{\left| {\tr\left( {{\bf{C}}_{lm}^l{\bf{R}}_{E,T}^l} \right)} \right|^2}.
\end{align}
Similarly, we have
\begin{multline} \label{eq:hEihmhEj_1}
\sum\limits_{r = 1}^{{N_e}} {\sum\limits_{t = 1}^{{N_e}} {{{\left( {{\bf{h}}_{E,i}^l} \right)}^H}{\bf{C}}_{lm}^l{\bf{h}}_{E,r}^l} } {\left( {{\bf{h}}_{E,t}^l} \right)^H}{\left( {{\bf{C}}_{lm}^l} \right)^H}{\bf{h}}_{E,j}^l \\
 \mathop \to \limits^{{N_t} \to \infty } \sum\limits_{r = 1}^{{N_e}} {{{\left\{ {{\bf{R}}_{E,R}^l} \right\}}_{ir}}} \sum\limits_{r = 1}^{{N_e}} {{{\left\{ {{\bf{R}}_{E,R}^l} \right\}}_{rj}}} {\left| {\tr\left( {{\bf{C}}_{lm}^l{\bf{R}}_{E,T}^l} \right)} \right|^2}
\end{multline}
\begin{multline} \label{eq:hEihmhEj_2}
{\left( {{\bf{h}}_{E,i}^l} \right)^H}{\bf{C}}_{lm}^l {\bf{\Omega }}_m {\bf{n}} {{\bf{n}}^H} {\bf{\Omega }}_m^H {\left( {{\bf{C}}_{lm}^l} \right)^H}{\bf{h}}_{E,j}^l  \\
 \mathop \to \limits^{{N_t} \to \infty } {N_0}\tau {\left\{ {{\bf{R}}_{E,R}^l} \right\}_{ij}} \tr\left( {{\bf{R}}_{E,T}^l{\bf{C}}_{lm}^l{{\left( {{\bf{C}}_{lm}^l} \right)}^H}} \right).
\end{multline}

Combining (\ref{eq:Q}), (\ref{eq:Q_l})--(\ref{eq:hEihmhEj_2}), we have
\begin{align} \label{eq:Q_Q_asy}
{\bf{Q}} \mathop \to \limits^{{N_t} \to \infty } \mathbf{Q}_{\rm asy}.
\end{align}
Based on (\ref{eq:Q_Q_asy}), we obtain
\begin{multline} \label{eq:w0m_HE}
{ {{{\bf{w}}_{0m}^H}}}{\bf{H}}_E^0{{\bf{Q}}^{ - 1}}{\left( {{\bf{H}}_E^0} \right)^H}{{\bf{w}}_{0m}} \\
\mathop \to \limits^{{N_t} \to \infty } \frac{1}{{{{\left\| {\widehat {\bf{h}}_{0m}^0} \right\|}^2}}}{\left( {\widehat {\bf{h}}_{0m}^0} \right)^H}{\bf{H}}_E^0{\bf{Q}}_{\rm asy}^{ - 1}{\left( {{\bf{H}}_E^0} \right)^H}\widehat {\bf{h}}_{0m}^0,
\end{multline}
which can be further simplified using
\begin{align} \label{eq:h0mh0m_he}
\frac{1}{{{N_t}}}{\left\| {\widehat {\bf{h}}_{0m}^0} \right\|^2} \mathop \to \limits^{{N_t} \to \infty } \frac{1}{{{N_t}}}\tr\left( {\widehat {\bf{R}}_{0m}^0} \right)
\end{align}
and
\begin{align} \label{eq:h0m_Q_asy_H_E}
& {\left( {\widehat {\bf{h}}_{0m}^0} \right)^H}{\bf{H}}_E^0{\bf{Q}}_{{\rm{asy}}}^{ - 1}{\left( {{\bf{H}}_E^0} \right)^H}\widehat {\bf{h}}_{0m}^0 \nonumber  \\
 &= \sum\limits_{i = 1}^{{N_e}} {\sum\limits_{j = 1}^{{N_e}} {{{\left\{ {{\bf{Q}}_{{\rm{asy}}}^{ - 1}} \right\}}_{ij}}} {{\left( {\widehat {\bf{h}}_{0m}^0} \right)}^H}\left( {{\bf{h}}_{E,i}^0} \right)} {\left( {{\bf{h}}_{E,j}^0} \right)^H}\widehat {\bf{h}}_{0m}^0 \nonumber \\
 & = \sum\limits_{i = 1}^{{N_e}} {\sum\limits_{j = 1}^{{N_e}} {{{\left\{ {{\bf{Q}}_{{\rm{asy}}}^{ - 1}} \right\}}_{ij}}} {{\left( {{\bf{h}}_{E,i}^0} \right)}^H}\widehat {\bf{h}}_{0m}^0} {\left( {\widehat {\bf{h}}_{0m}^0} \right)^H}{\bf{h}}_{E,j}^0 \nonumber  \\
 & \mathop \to \limits^{{N_t} \to \infty }  \sum\limits_{i = 1}^{{N_e}} {\sum\limits_{j = 1}^{{N_e}} {{{\left\{ {{\bf{Q}}_{{\rm{asy}}}^{ - 1}} \right\}}_{ij}}\eta _{ij}^0} },
\end{align}
which was obtained by following a similar approach as was used to obtain (\ref{eq:Q_Q_asy}). Substituting  (\ref{eq:h0mh0m_he}) and (\ref{eq:h0m_Q_asy_H_E}) into (\ref{eq:C_eve_1}) completes the proof.

{\bl
\section{Proof of Theorem \ref{prop:all_transmit}}\label{proof:prop:all_transmit}
When $p = \frac{1}{K}$,  $\frac{{1 + {\rm SINR}_{0m, \, {\rm asy} }}}{{1 + {\rm SINR}_{\rm{eve},\, {\rm asy}}}}$ can be simplified to
\begin{align} \label{eq:SINR_ratio_app_snr}
 \frac{{1 + {\rm SINR}_{0m, \, {\rm asy} }}}{{1 + {\rm SINR}_{\rm{eve},\, {\rm asy}}}} =  \frac{{\gamma {\theta _{b,p}} + \gamma {\theta _m} + 1}}{{\left( {\gamma {\theta _{b,p}} + 1} \right)\left( {\gamma {\tilde{\theta} _{e}} + 1} \right)}}.
\end{align}
The derivative of (\ref{eq:SINR_ratio_app_snr}) with respect to $\gamma$ is given by
\begin{multline} \label{eq:SINR_ratio_der_app_snr_one}
\frac{{d\left( \frac{{\gamma {\theta _{b,p}} + \gamma {\theta _m} + 1}}{{\left( {\gamma {\theta _{b,p}} + 1} \right)\left( {\gamma {\tilde{\theta} _{e}} + 1} \right)}}  \right)}}{{d\gamma }} \\  = \frac{{ - \left( {{\theta _{b,p}} + {\theta _m}} \right){\theta _{b,p}}{\tilde{\theta} _{e}}{\gamma ^2} - 2\gamma {\theta _{b,p}}{\theta _m} + {\theta _m} - {\tilde{\theta} _{e}}}}{{{{\left( {\gamma {\theta _{b,p}} + 1} \right)}^2}{{\left( {\gamma {\tilde{\theta} _{e}} + 1} \right)}^2}}}.
\end{multline}
From (\ref{eq:theta_m}) and (\ref{eq:theta_bp}),
we observe that $\theta _m \geq 0$ and $\theta _{b,p} \geq 0$.

Next, we prove  $\tilde{\theta} _{e} \geq 0$.
From (\ref{eq:eta_prop}), we obtain
\begin{align} \label{eq:eta_ii}
& \sum\limits_{i = 1}^{{N_e}} \eta _{ii}^0  \nonumber \\
& = {\tau ^2} \sum\limits_{i = 1}^{{N_e}}{\left\{ {{\bf{R}}_{E,R}^0} \right\}_{ii}}\sum\limits_{t = 0}^L {{P_{tm}}\tr\left( {{\bf{R}}_{E,T}^0{\bf{C}}_{0m}^0{\bf{R}}_{tm}^0{{\left( {{\bf{C}}_{0m}^0} \right)}^H}} \right)} \nonumber \\
& + {\tau ^2}\frac{{{P_E}}}{{{N_e}}} {\left| {\tr\left( {{\bf{C}}_{0m}^0{\bf{R}}_{E,T}^0} \right)} \right|^2}
 \sum\limits_{i = 1}^{{N_e}} \sum\limits_{r = 1}^{{N_e}} {{{\left\{ {{\bf{R}}_{E,R}^0} \right\}}_{ir}}} \sum\limits_{r = 1}^{{N_e}} {{{\left\{ {{\bf{R}}_{E,R}^0} \right\}}_{ri}}} \nonumber \\
 & + {N_0}\tau \tr\left( {{\bf{R}}_{E,T}^0{\bf{C}}_{0m}^0{{\left( {{\bf{C}}_{0m}^0} \right)}^H}} \right)  \sum\limits_{i = 1}^{{N_e}} {\left\{ {{\bf{R}}_{E,R}^0} \right\}_{ii}}.
\end{align}
Since ${{\bf{R}}_{E,R}^0}$
is a receive correlation matrix, it is a Hermitian positive-semidefinite matrix.
Therefore, we have ${{{\left\{ {{\bf{R}}_{E,R}^0} \right\}}_{ii}}} \geq 0$.
As a result, we obtain
\begin{align} \label{eq:eta_13}
& {\tau ^2}{\left\{ {{\bf{R}}_{E,R}^0} \right\}_{ii}}\sum\limits_{t = 0}^L {{P_{tm}}\tr\left( {{\bf{R}}_{E,T}^0{\bf{C}}_{0m}^0{\bf{R}}_{tm}^0{{\left( {{\bf{C}}_{0m}^0} \right)}^H}} \right)} \geq 0 \\
 &  {N_0}\tau {\left\{ {{\bf{R}}_{E,R}^0} \right\}_{ii}}\tr\left( {{\bf{R}}_{E,T}^0{\bf{C}}_{0m}^0{{\left( {{\bf{C}}_{0m}^0} \right)}^H}} \right) \geq 0.
\end{align}
Also, we have
\begin{align}
 \sum\limits_{i = 1}^{{N_e}} \sum\limits_{r = 1}^{{N_e}} {{{\left\{ {{\bf{R}}_{E,R}^0} \right\}}_{ir}}} \sum\limits_{r = 1}^{{N_e}} {{{\left\{ {{\bf{R}}_{E,R}^0} \right\}}_{ri}}} & = {{\bf{1}}_{{N_e}}^H} {{\bf{R}}_{E,R}^0} {{\bf{R}}_{E,R}^0} {{\bf{1}}_{{N_e}}} \nonumber \\
 & =  {\left({{\bf{R}}_{E,R}^0}{\bf{1}}_{{N_e}}\right)^H} {{\bf{R}}_{E,R}^0} {{\bf{1}}_{{N_e}}} \nonumber c\\
 & \geq 0  \label{eq:eta_2}
\end{align}
From (\ref{eq:gamma_th_2}), (\ref{eq:eta_ii})--(\ref{eq:eta_2}), we obtain  $\tilde{\theta} _{e} \geq 0$.
As a result, the equation
\begin{align} \label{eq:SINR_ratio_der_app_eqn}
{ - \left( {{\theta _{b,p}} + {\theta _m}} \right){\theta _{b,p}}{\tilde{\theta} _{e}}{\gamma ^2} - 2\gamma {\theta _{b,p}}{\theta _m} + {\theta _m} - {\tilde{\theta} _{e}}} = 0
\end{align}
has  at least one non-positive root. We assume that $\gamma_1$ and $\gamma_2$, $\gamma_1 < \gamma_2$ are the two roots of (\ref{eq:SINR_ratio_der_app_eqn}) and $\gamma_1 \leq 0$. Since $- \left( {{\theta _{b,p}}
+ {\theta _m}} \right){\theta _{b,p}}{\tilde{\theta} _{e}} <0$, we have
  \begin{enumerate}
\item  When ${\gamma _1} \leq \gamma \leq {\gamma _2}$, ${ - \left( {{\theta _{b,p}} + {\theta _m}} \right){\theta _{b,p}}{\tilde{\theta} _{e}}{\gamma ^2} - 2\gamma {\theta _{b,p}}{\theta _m}
 + {\theta _m} - {\tilde{\theta} _{e}}} \geq 0$.

\item  When $\gamma <{\gamma _1} $ or $ \gamma > {\gamma _2}$, ${ - \left( {{\theta _{b,p}} + {\theta _m}} \right){\theta _{b,p}}{\tilde{\theta} _{e}}{\gamma ^2} - 2\gamma {\theta _{b,p}}{\theta _m}
+ {\theta _m} - {\tilde{\theta} _{e}}} < 0 $.
  \end{enumerate}

Since $\gamma_1 \leq 0$, $\gamma <{\gamma _1} $ can not hold.
When $ \gamma > {\gamma _2}$,  we know from (\ref{eq:SINR_ratio_der_app_snr_one})
 \begin{align} \label{eq:SINR_ratio_der_app_snr_one_2}
\frac{{d\left( \frac{{\gamma {\theta _{b,p}} + \gamma {\theta _m} + 1}}{{\left( {\gamma {\theta _{b,p}} + 1} \right)\left( {\gamma {\tilde{\theta} _{e}} + 1} \right)}}  \right)}}{{d\gamma }} < 0.
\end{align}
Finally,  setting $\gamma_{\rm th} = {\gamma _2}$ in Theorem \ref{prop:all_transmit} completes the proof.
}

\section{Proof of Theorem \ref{prop:sec_orth}}\label{proof:prop:sec_orth}
Since ${{\bf{R}}_{tm}^l}$ and ${{\bf{R}}_{E,T}^{l}}$ are positive semi-definite correlation matrices,
$\sum\nolimits_{t = 0}^L  \tr\left( {{\bf{R}}_{tm}^l}{{\bf{R}}_{E,T}^{l}}\right) = 0$ is equivalent to
$\sum\nolimits_{t = 0}^L  {{\bf{R}}_{tm}^l}{{\bf{R}}_{E,T}^{l}} =  \mathbf{0}$. For $\sum\nolimits_{t = 0}^L  {{\bf{R}}_{tm}^l}{{\bf{R}}_{E,T}^{l}} = \mathbf{0}$,
we have (\ref{eq:corr_orth}) at the top of the next page, where $\mathop  = \limits^{(a)}$ follows from the matrix inversion lemma \cite{Bernstein2005}.
 \begin{figure*}[!ht]
\begin{align} \label{eq:corr_orth}
\widehat {\bf{R}}_{lm}^l &= {P_{lm}}\tau {\bf{R}}_{lm}^l{\left( {{N_0}{{\bf{I}}_{{N_t}}} + \tau \left( {\sum\limits_{t = 0}^L {{P_{tk}}{\bf{R}}_{tk}^l}  + {P_E}r_{E,R}^l{\bf{R}}_{E,T}^l} \right)} \right)^{ - 1}}{\bf{R}}_{lm}^l
\nonumber \\
 &= {P_{lm}}\tau {\bf{R}}_{lm}^l{\left( {\left( {\sqrt {{N_0}} {{\bf{I}}_{{N_t}}} + \frac{\tau }{{\sqrt {{N_0}} }}\sum\limits_{t = 0}^L {{P_{tk}}{\bf{R}}_{tk}^l} } \right)\left( {\sqrt {{N_0}} {{\bf{I}}_{{N_t}}} + \frac{{\tau {P_E}r_{E,R}^l}}{{\sqrt {{N_0}} }}{\bf{R}}_{E,T}^l} \right)} \right)^{ - 1}}{\bf{R}}_{lm}^l \nonumber \\
& = {P_{lm}}\tau {\bf{R}}_{lm}^l{\left( {\sqrt {{N_0}} {{\bf{I}}_{{N_t}}} + \frac{{\tau {P_E}r_{E,R}^l}}{{\sqrt {{N_0}} }}{\bf{R}}_{E,T}^l} \right)^{ - 1}}{\left( {\sqrt {{N_0}} {{\bf{I}}_{{N_t}}} + \frac{\tau }{{\sqrt {{N_0}} }}\sum\limits_{t = 0}^L {{P_{tk}}{\bf{R}}_{tk}^l} } \right)^{ - 1}}{\bf{R}}_{lm}^l \nonumber \\
&\mathop  =  \limits^{\left( a \right)} {P_{lm}}\tau {\bf{R}}_{lm}^l\left( {\frac{1}{{\sqrt {{N_0}} }}{{\bf{I}}_{{N_t}}} \! -  \! \frac{{\tau {P_E}r_{E,R}^l}}{{\sqrt {N_0^3} }}{\bf{R}}_{E,T}^l{{\left( {{{\bf{I}}_{{N_t}}}  \!+  \! \frac{{\tau {P_E}r_{E,R}^l}}{{{N_0}}}{\bf{R}}_{E,T}^l} \right)}^{ - 1}}} \right) \! {\left( {\sqrt {{N_0}} {{\bf{I}}_{{N_t}}}  \! +  \!\frac{\tau }{{\sqrt {{N_0}} }}\sum\limits_{t = 0}^L {{P_{tk}}{\bf{R}}_{tk}^l} } \right)^{ - 1}}{\bf{R}}_{lm}^l \nonumber \\
 &= {P_{lm}}\tau {\bf{R}}_{lm}^l{\left( {{{\bf{I}}_{{N_t}}} + \tau \sum\limits_{t = 0}^L {{P_{tk}}{\bf{R}}_{tk}^l} } \right)^{ - 1}}{\bf{R}}_{lm}^l
\end{align}
 \hrulefill
\vspace*{4pt}
\end{figure*}
Similarly, we can prove that when $\sum\nolimits_{t = 0}^L  {{\bf{R}}_{tm}^l}{{\bf{R}}_E^{l}} = \mathbf{0}$,  ${\bf{C}}_{lm}^l$ reduces to  ${\bf{C}}_{lm,\rm{orth}}^l$.

Also, when $\sum\nolimits_{t = 0}^L {{\bf{R}}_{tm}^l}{{\bf{R}}_{E,T}^{l}} = \mathbf{0}$, we have
\begin{align} \label{eq:Clm_RE_orth}
{\bf{R}}_{E,T}^l{\left( {{\bf{C}}_{lm}^l} \right)^H} = {\bf{R}}_{E,T}^l{\bf{R}}_{lm}^l{\left( {{{\bf{I}}_{{N_t}}} + \tau \sum\limits_{t = 0}^L {{P_{tk}}{\bf{R}}_{tk}^l} } \right)^{ - 1}}  = \mathbf{0},
\end{align}
and as a result,
\begin{align}\label{eq:Clm_RE_orth_2}
& \tr\left( {{\bf{R}}_{0m}^l{\bf{C}}_{lm}^l{\bf{R}}_{E,T}^l{{\left( {{\bf{C}}_{lm}^l} \right)}^H}} \right) = 0, \\
 & \tr\left( {{\bf{R}}_{E,T}^l{\bf{C}}_{lm}^l{{\left( {{\bf{C}}_{lm}^l} \right)}^H}} \right) = 0, \\
& \tr\left( {{\bf{R}}_{E,T}^l{\bf{C}}_{lm}^l} \right) = 0.
\end{align}
Substituting (\ref{eq:Clm_RE_orth_2}) into (\ref{eq:eta_prop}), yields
\begin{align}\label{eq:eta_1}
\eta _{ij}^l = 0.
\end{align}
The proof is completed by substituting (\ref{eq:corr_orth})--(\ref{eq:eta_1}) into (\ref{eq:asy_sec}).

\section{Proof of Theorem \ref{prop:opt_power}}\label{proof:prop:opt_power}
When $N_e = 1$, (\ref{eq:e}) reduces to
 \begin{align}
{\rm SINR}_{\rm eve,\, asy} &= \frac{{p\gamma }{\theta _{e,e}}}{{q\gamma {\theta _{e,q}} + 1}}.
\end{align}
Hence, maximizing $R_{\sec ,\, {\rm asy}}$ in (\ref{eq:asy_sec}) is equivalent to maximizing
\begin{align} \label{eq:SINR_ratio_app}
\frac{{1 + {\rm SINR}_{0m, \, {\rm asy} }}}{{1 + {\rm SINR}_{\rm{eve},\, {\rm asy}}}} = \frac{{ {a_1}{p^2} + {b_1}p + {c_1} }}{{{a_2}{p^2} + {b_2}p + {c_2}}}.
\end{align}
By taking the derivative of (\ref{eq:SINR_ratio_app}) with respect to $p$, we obtain
\begin{multline} \label{eq:SINR_ratio_der_app}
\frac{{d\left( {\frac{{{a_1}{p^2} + {b_1}p + {c_1}}}{{{a_2}{p^2} + {b_2}p + {c_2}}}} \right)}}{{dp}} \\
 = \frac{{\left( {{a_1}{b_2} - {a_2}{b_1}} \right){p^2} + 2\left( {{a_1}{c_2} - {a_2}{c_1}} \right)p + {b_1}{c_2} - {b_2}{c_1}}}{{{{\left( {{a_2}{p^2} + {b_2}p + {c_2}} \right)}^2}}}.
\end{multline}
The proof is completed by setting (\ref{eq:SINR_ratio_der_app}) to zero and finding the solution.

\section{Proof of Theorem \ref{prop:unified}} \label{proof:prop:unified}
Based on Theorems \ref{prop:sec_rate_mul} and \ref{prop:opt_power},
the optimal asymptotic secrecy rate for the MF-AN design is given by
\begin{align} \label{eq:sec_asy_app_u}
R_{\sec ,\, {\rm asy}}^* = \log_2 \frac{{a_3{\gamma ^2} + b_3\gamma  + 1}}{{a_4 {\gamma ^2} + b_4\gamma  + 1}} .
\end{align}

Following a similar approach as in Appendix \ref{proof:prop:sec_rate_mul},
we obtain the asymptotic secrecy rate of the NS design as
\begin{align} \label{eq:sec_asy_null_app_u}
R_{\sec, \, {\rm asy, \, null}} = \log_2  \frac{{ a_5\gamma + 1}}{{ a_6\gamma + 1}}.
\end{align}

Performing some simplifications, $R_{\sec ,\, {\rm asy}}^*  \gtreqqless R_{\sec, \, {\rm asy, \, null}}$ is
equivalent to
\begin{align} \label{eq:sec_asy_u_condition}
{a_7}{\gamma ^3} + {b_7}{\gamma ^2} + {c_7}\gamma \gtreqqless 0.
\end{align}

Then, Theorem \ref{prop:unified} is obtained by finding the feasible region of (\ref{eq:sec_asy_u_condition}).

{\bl

\section{Proof of Theorem \ref{prop:sec_cond_single_user}} \label{proof:prop:sec_cond_single_user}
Following a similar approach as in  Appendix \ref{proof:prop:sec_rate_mul}, when $L = 0$, $K = 1$, $N_e = 1$, and
$N_t \rightarrow \infty$,  ${\rm SINR}_{01,\,{\rm{asy}}}$ in (\ref{eq:b}) becomes (\ref{eq:sinr_01_asy}),
which is given  at the top of the next page.
 \begin{figure*}[!ht]
\begin{align} \label{eq:sinr_01_asy}
 {\rm SINR}_{01,\,{\rm{asy}}}  = \frac{{p\gamma \left( {{\tr^2}\left( {\widehat {\bf{R}}_{01}^0} \right) + \tr\left( {\left( {{\bf{R}}_{01}^0 - \widehat {\bf{R}}_{01}^0} \right)\widehat {\bf{R}}_{01}^0} \right)} \right)}}{{q\gamma \tr\left( {\left( {{\bf{R}}_{01}^0 - \widehat {\bf{R}}_{01}^0} \right)\widehat {\bf{R}}_{01}^0} \right) + q\gamma \tr\left( {\widehat {\bf{R}}_{01}^0} \right) \tr\left( {{\bf{R}}_{01}^0 - \widehat {\bf{R}}_{01}^0} \right) + \tr\left( {\widehat {\bf{R}}_{01}^0} \right)}}.
\end{align}
 \hrulefill
\vspace*{4pt}
\end{figure*}

Also,  ${\rm{SIN}}{{\rm{R}}_{{\rm eve},\,{\rm{asy}}}}$ in (\ref{eq:e}) becomes
\begin{align} \label{eq:sinr_e_asy}
 {\rm{SIN}}{{\rm{R}}_{{\rm eve},\,{\rm{asy}}}} = \frac{{p\gamma \Lambda }}{{q\gamma \left( {\tr\left( {{\bf{R}}_{E,T}^0} \right) \tr\left( {\widehat {\bf{R}}_{01}^0} \right) - \Lambda } \right) + \tr\left( {\widehat {\bf{R}}_{01}^0} \right)}}.
\end{align}

For secure communication, we require ${\rm SINR}_{01,\,{\rm{asy}}} > {\rm{SIN}}{{\rm{R}}_{{\rm eve},\,{\rm{asy}}}}$,
which is equivalent to (\ref{eq:condition_single}) at the top of the next page.
 \begin{figure*}[!ht]
\begin{align} \label{eq:condition_single}
\frac{{p\gamma \left( {\tr^2\left( {\widehat {\bf{R}}_{01}^0} \right) + \tr\left( {\left( {{\bf{R}}_{01}^0 - \widehat {\bf{R}}_{01}^0} \right)\widehat {\bf{R}}_{01}^0} \right)} \right)}}{{q\gamma \tr\left( {\left( {{\bf{R}}_{01}^0 - \widehat {\bf{R}}_{01}^0} \right)\widehat {\bf{R}}_{01}^0} \right) + q\gamma \tr\left( {\widehat {\bf{R}}_{01}^0} \right) \tr\left( {{\bf{R}}_{01}^0 - \widehat {\bf{R}}_{01}^0} \right) + \tr\left( {\widehat {\bf{R}}_{01}^0} \right)}} > \frac{{p\gamma \Lambda }}{{q\gamma \left( {\tr\left( {{\bf{R}}_{E,T}^0} \right)\tr\left( {\widehat {\bf{R}}_{01}^0} \right) - \Lambda } \right) + \tr\left( {\widehat {\bf{R}}_{01}^0} \right)}}.
\end{align}
 \hrulefill
\vspace*{4pt}
\end{figure*}

When $N_t \rightarrow \infty$ and $N_e = 1$, we have
\begin{align} \label{eq:h_E_h010}
{\left\| {{\bf{h}}_{E,1}^0} \right\|^2}{\left\| {\widehat {\bf{h}}_{01}^0} \right\|^2} & \mathop  \to \limits^{{N_t} \to \infty }  \tr\left( {{\bf{R}}_{E,T}^0} \right)\tr\left( {\widehat {\bf{R}}_{01}^0} \right) \\
\left({\bf{h}}_{E,1}^0\right)^H {\widehat {\bf{h}}_{01}^0}  \left(\widehat {\bf{h}}_{01}^0\right)^H {{\bf{h}}_{E,1}^0} & \mathop  \to \limits^{{N_t} \to \infty } \Lambda.
\end{align}
Exploiting the Cauchy-Schwarz inequality, we obtain
\begin{align} \label{eq:Cauchy-Schwarz}
\left({\bf{h}}_{E,1}^0\right)^H {\widehat {\bf{h}}_{01}^0}  \left(\widehat {\bf{h}}_{01}^0\right)^H {{\bf{h}}_{E,1}^0} \leq {\left\| {{\bf{h}}_{E,1}^0} \right\|^2}{\left\| {\widehat {\bf{h}}_{01}^0} \right\|^2}.
\end{align}
Therefore, we have
\begin{align}\label{eq:RET_R010}
 {\tr\left( {{\bf{R}}_{E,T}^0} \right)\tr\left( {\widehat {\bf{R}}_{01}^0} \right) - \Lambda} \geq 0.
\end{align}
Moreover, based on (\ref{eq:R0m0_est}), we have
\begin{align}
& {\bf{R}}_{01}^0 - \widehat {\bf{R}}_{01}^0 \nonumber \\
 & = {\bf{R}}_{01}^0{\left( {{N_0}{{\bf{I}}_{{N_t}}} + \tau \left( {{P_{01}}{\bf{R}}_{01}^0 + {P_E}{\bf{R}}_{E,T}^0} \right)} \right)^{ - 1}}  \nonumber  \\
 & \hspace{1cm} \times \left( {{N_0}{{\bf{I}}_{{N_t}}} + \tau {P_E}{\bf{R}}_{E,T}^0} \right) \nonumber \\
 & = {\bf{R}}_{01}^0{\left( {{{\bf{I}}_{{N_t}}} + \tau {P_{01}}{\bf{R}}_{01}^0{{\left( {{N_0}{{\bf{I}}_{{N_t}}} + \tau {P_E}{\bf{R}}_{E,T}^0} \right)}^{ - 1}}} \right)^{ - 1}}. \label{eq:R010_sub_2}
\end{align}

Since ${{\bf{R}}_{01}^0}$ and ${{\bf{R}}_{E,T}^0}$ are transmit correlation matrices, they are Hermitian positive-semidefinite matrices.
From (\ref{eq:R0m0_est}) and (\ref{eq:R010_sub_2}), we observe
\begin{align}\label{eq:R010_est_tr}
\tr\left( {\widehat {\bf{R}}_{01}^0} \right) > 0,  \quad  \tr\left( {{\bf{R}}_{01}^0 - \widehat {\bf{R}}_{01}^0} \right) > 0.
\end{align}
Eqs. (\ref{eq:RET_R010}) and (\ref{eq:R010_est_tr}) indicate that the denominators of both the left hand and the right hand terms in (\ref{eq:condition_single})
are positive. Then, simplifying (\ref{eq:condition_single}), we obtain (\ref{eq:condition_single_theo}). This completes the proof.

\section{Proof of Theorem \ref{prop:sec_cond_single_user_iid}} \label{proof:prop:sec_cond_single_user_iid}
Setting ${\bf{R}}_{01}^0= {\beta _{01}}{{\bf{I}}_{{N_t}}}$ and ${{\bf{R}}_{E,T}^0} = {\beta _E}{{\bf{I}}_{{N_t}}}$, we have from (\ref{eq:Rnmn_est})
\begin{align}\label{eq:simple_exp_iid}
&\widehat {\bf{R}}_{01}^0 = \frac{{\tau {P_{01}}\beta _{01}^2}}{{{N_0} + \tau \left( {{P_{01}}{\beta _{01}} + {P_E}{\beta _E}} \right)}}{{\bf{I}}_{{N_t}}} \\
&\tr\left(\widehat {\bf{R}}_{01}^0 \right)  = \frac{{\tau {P_{01}}\beta _{01}^2{N_t}}}{{{N_0} + \tau \left( {{P_{01}}{\beta _{01}} + {P_E}{\beta _E}} \right)}} \\
&\tr\left( \left({\bf{R}}_{01}^0 - \widehat {\bf{R}}_{01}^0 \right) \widehat {\bf{R}}_{01}^0 \right)  = \frac{{\tau {P_{01}}\beta _{01}^3\left( {{N_0} + \tau {P_E}{\beta _E}} \right){N_t}}}{{{{\left( {{N_0} + \tau \left( {{P_{01}}{\beta _{01}} + {P_E}{\beta _E}} \right)} \right)}^2}}}  \\
 &\tr \left({\bf{R}}_{01}^0 - \widehat {\bf{R}}_{01}^0 \right)  = \frac{{{\beta _{01}}\left( {{N_0} + \tau {P_E}{\beta _E}} \right){N_t}}}{{{N_0} + \tau \left( {{P_{01}}{\beta _{01}} + {P_E}{\beta _E}} \right)}}  \\
&\Lambda  = \frac{{\tau {P_{01}}\beta _{01}^2{\beta _E}{N_t}}}{{{{\left( {{N_0} + \tau \left( {{P_{01}}{\beta _{01}} + {P_E}{\beta _E}} \right)} \right)}^2}}} \nonumber \\
& \hspace{3cm} \times \left( {\tau {P_{01}}{\beta _{01}} + \tau {P_E}{\beta _E}{N_t} + {N_0}} \right) \\
&\tr\left( {{{\bf{R}}_{E,T}^0}} \right)  = {\beta _E}{N_t}.  \label{eq:ret_iid}
\end{align}

Substituting (\ref{eq:simple_exp_iid})--(\ref{eq:ret_iid}) into $\eta_1$ in (\ref{eq:coff}), recalling $N_t \rightarrow\infty$,
and simplifying, we obtain
\begin{align}\label{eq:simple_exp_iid_2}
{\eta _1} = N_t^4 \left( {\tau {P_{01}}{\beta _{01}} + \tau {P_E}{\beta _E} + {N_0}} \right)\left( {{P_{01}}{\beta _{01}} - {P_E}{\beta _E}} \right).
\end{align}
This completes the proof.

}


\end{document}